\documentclass[twoside,leqno,twocolumn]{article}
\usepackage{ltexpprt}

\usepackage[ruled, vlined, linesnumbered, nofillcomment]{algorithm2e}

\usepackage[english]{babel}
\usepackage{amssymb}
\usepackage{amsmath}
\usepackage{amsthm}
\usepackage{mathrsfs}
\usepackage{subfigure}
\usepackage{booktabs}
\usepackage{cite}
\usepackage{url}
\usepackage{graphicx}
\usepackage{mdwlist}
\usepackage{tikz}
\usetikzlibrary{positioning,fit,calc,external}
\usepackage{pgfplots}
\usepackage{thmtools}
\usepackage[numbers]{natbib}
\newcommand{\set}[1]{\left\{#1\right\}}
\newcommand{\pr}[1]{\left(#1\right)}
\newcommand{\fpr}[1]{\mathopen{}\left(#1\right)}
\newcommand{\spr}[1]{\left[#1\right]}

\newcommand{\abs}[1]{{\left|#1\right|}}

\newcommand{\enpr}[2]{\pr{#1 ,\ldots , #2}}

\newcommand{\define}{\leftarrow}

\newcommand{\bigO}[1]{\dispfunc{\mathcal{O}}{#1}}

\DeclareRobustCommand{\dispfunc}[2]{%
  \ensuremath{%
  \ifthenelse{\equal{#2}{}}%
    {\mathit{#1}}%
    {\mathit{#1}\fpr{#2}}}}

\newcommand{\pen}[1]{\dispfunc{\tau}{#1}}
\newcommand{\scoreexp}[1]{\dispfunc{q_{exp}}{#1}}
\newcommand{\scoregeo}[1]{\dispfunc{q_{geo}}{#1}}

\newcommand{\pexp}[1]{\dispfunc{p_{exp}}{#1}}
\newcommand{\pgeo}[1]{\dispfunc{p_{geo}}{#1}}

\newcommand{\boundprb}{\textsc{BndBurst}\xspace}
\newcommand{\brstexpprb}{\textsc{Exp}\xspace}
\newcommand{\brstgeoprb}{\textsc{Geo}\xspace}

\newcommand{\algdp}{\textit{Viterbi}\xspace}
\newcommand{\alggeo}{\textit{ApproxGeo}\xspace}
\newcommand{\alggeoalpha}{\textit{GeoAlpha}\xspace}
\newcommand{\algexp}{\textit{ApproxExp}\xspace}
\newcommand{\algexpalpha}{\textit{ExpAlpha}\xspace}
\newcommand{\algmeanexp}{\textit{ExpMean}\xspace}
\newcommand{\algmeangeo}{\textit{GeoMean}\xspace}

\newcommand{\dtname}[1]{\textsl{#1}}

\newtheorem{lemma}{Lemma}[section]
\newtheorem{proposition}{Proposition}[section]

\newtheorem{problem}{Problem}[section]

\newtheorem{example}{Example}[section]

\SetKwComment{tcpas}{\{}{\}}
\SetCommentSty{textnormal}
\SetArgSty{textnormal}
\SetKw{False}{false}
\SetKw{True}{true}
\SetKw{Null}{null}
\SetKwInOut{Output}{output}
\SetKwInOut{Input}{input}
\SetKw{AND}{and}
\SetKw{OR}{or}
\SetKw{Break}{break}

\pgfdeclarelayer{background}
\pgfdeclarelayer{foreground}
\pgfsetlayers{background,main,foreground}

\definecolor{yafaxiscolor}{rgb}{0.3, 0.3, 0.3}

\definecolor{yafcolor1}{rgb}{0.4, 0.165, 0.553}
\definecolor{yafcolor2}{rgb}{0.949, 0.482, 0.216}
\definecolor{yafcolor3}{rgb}{0.47, 0.549, 0.306}
\definecolor{yafcolor4}{rgb}{0.925, 0.165, 0.224}
\definecolor{yafcolor5}{rgb}{0.141, 0.345, 0.643}
\definecolor{yafcolor6}{rgb}{0.965, 0.933, 0.267}
\definecolor{yafcolor7}{rgb}{0.627, 0.118, 0.165}
\definecolor{yafcolor8}{rgb}{0.878, 0.475, 0.686}

\tikzstyle{exnode} = [inner sep = 1pt]
\tikzstyle{labnode} = [sloped, text = black, font = \scriptsize, inner sep = 1pt]
\tikzstyle{exedge} = [yafcolor5, draw, thick, >=latex, ->]
\tikzstyle{exedge2} = [yafcolor2, draw, thick, >=latex, ->]

\tikzstyle{interval} = [yafcolor1, very thick]

\newlength{\intervalwidth}
\newcommand{\drawinterval}[3]{
\draw[#3] (#1) ++(0, \intervalwidth) -- (#1) -- (#2) -- ++(0, \intervalwidth);
}

\newlength{\yafaxispad}
\newlength{\yaftlpad}
\newlength{\yaflabelpad}
\newlength{\yafaxiswidth}
\newlength{\yafticklen}
\makeatletter
\def\pgfplots@drawtickgridlines@INSTALLCLIP@onorientedsurf#1{}
\makeatother

\newcommand{\yafdrawxaxis}[2]{
	\pgfplotstransformcoordinatex{#1}\let\xmincoord=\pgfmathresult 
	\pgfplotstransformcoordinatex{#2}\let\xmaxcoord=\pgfmathresult 
	\pgfsetlinewidth{\yafaxiswidth} 
	\pgfsetcolor{yafaxiscolor}
	\pgfpathmoveto{\pgfpointadd{\pgfpointadd{\pgfplotspointrelaxisxy{0}{0}}{\pgfqpointxy{\xmincoord}{0}}}{\pgfqpoint{-0.5\yafaxiswidth}{\yafaxispad}}}
	\pgfpathlineto{\pgfpointadd{\pgfpointadd{\pgfplotspointrelaxisxy{0}{0}}{\pgfqpointxy{\xmaxcoord}{0}}}{\pgfqpoint{0.5\yafaxiswidth}{\yafaxispad}}}
	\pgfusepath{stroke}

}
\newcommand{\yafdrawyaxis}[2]{
	\pgfplotstransformcoordinatey{#1}\let\ymincoord=\pgfmathresult 
	\pgfplotstransformcoordinatey{#2}\let\ymaxcoord=\pgfmathresult 
	\pgfsetlinewidth{\yafaxiswidth} 
	\pgfsetcolor{yafaxiscolor}
	\pgfpathmoveto{\pgfpointadd{\pgfpointadd{\pgfplotspointrelaxisxy{0}{0}}{\pgfqpointxy{0}{\ymincoord}}}{\pgfqpoint{\yafaxispad}{-0.5\yafaxiswidth}}}
	\pgfpathlineto{\pgfpointadd{\pgfpointadd{\pgfplotspointrelaxisxy{0}{0}}{\pgfqpointxy{0}{\ymaxcoord}}}{\pgfqpoint{\yafaxispad}{0.5\yafaxiswidth}}}
	\pgfusepath{stroke}
}

\newcommand{\yafdrawaxis}[4]{\yafdrawxaxis{#1}{#2}\yafdrawyaxis{#3}{#4}}

\pgfplotscreateplotcyclelist{yaf}{%
{yafcolor1,mark options={scale=0.75},mark=o}, 
{yafcolor2,mark options={scale=0.75},mark=square},
{yafcolor3,mark options={scale=0.75},mark=triangle},
{yafcolor4,mark options={scale=0.75},mark=o},
{yafcolor5,mark options={scale=0.75},mark=o},
{yafcolor6,mark options={scale=0.75},mark=o},
{yafcolor7,mark options={scale=0.75},mark=o},
{yafcolor8,mark options={scale=0.75},mark=o}} 

\pgfplotsset{axis y line=left, axis x line=bottom,
	tick align=outside,
	compat = 1.3,
	tickwidth=\yafticklen,
	clip = false,
	every axis title shift = 0pt,
    x axis line style= {-, line width = 0pt, opacity = 0},
    y axis line style= {-, line width = 0pt, opacity = 0},
    x tick style= {line width = \yafaxiswidth, color=yafaxiscolor, yshift = \yafaxispad},
    y tick style= {line width = \yafaxiswidth, color=yafaxiscolor, xshift = \yafaxispad},
    x tick label style = {font=\scriptsize, yshift = \yaftlpad},
    y tick label style = {font=\scriptsize, xshift = \yaftlpad},
    every axis y label/.style = {at = {(ticklabel cs:0.5)}, rotate=90, anchor=center, font=\scriptsize, yshift = -\yaflabelpad},
    every axis x label/.style = {at = {(ticklabel cs:0.5)}, anchor=center, font=\scriptsize, yshift = \yaflabelpad},
    x tick label style = {font=\scriptsize, yshift = 1pt},
    grid = major,
    major grid style  = {dash pattern = on 1pt off 3 pt},
	every axis plot post/.append style= {line width=\yafaxiswidth} ,
	legend cell align = left,
	legend style = {inner sep = 1pt, cells = {font=\scriptsize}},
	legend image code/.code={%
		\draw[mark repeat=2,mark phase=2,#1] 
		plot coordinates { (0cm,0cm) (0.15cm,0cm) (0.3cm,0cm) };%
	} 
}

\begin{document}

\title{Discovering bursts revisited:\\guaranteed optimization of the model parameters}

\author{Nikolaj Tatti \\ HIIT, Aalto University, Finland, \texttt{nikolaj.tatti@aalto.fi}}
\date{}

\maketitle

\begin{abstract} \small\baselineskip=9pt
One of the classic data mining
tasks is to discover bursts, time intervals, where events occur at
abnormally high rate.
In this paper we revisit Kleinberg's seminal work, where bursts are discovered
by using exponential distribution with
a varying rate parameter: the regions where it is more advantageous to set the
rate higher are deemed bursty.
The model depends on two parameters, the initial rate and the change rate.
The initial rate, that is, the rate that is used
when there are no burstiness was set to the average rate over the whole
sequence. The change rate is provided by the user.

We argue that these choices are suboptimal: it leads to worse likelihood,
and may lead to missing some existing bursts. We propose an alternative problem
setting, where the model parameters are selected by optimizing the likelihood of the
model. While this tweak is trivial from the problem definition point of view,
this changes the optimization problem greatly.
To solve the problem in practice,
we propose efficient ($1 + \epsilon$) approximation schemes. 
Finally, we demonstrate empirically that
with this setting we are able to discover bursts that would have otherwise be undetected.

\end{abstract}

\section{Introduction}\label{sec:intro}

Many natural phenomena occur unevenly over time, and one of the classic data mining
tasks is to discover bursts, time intervals, where events occur at
abnormally high rate. In this paper we revisit a seminal work by
\citet{kleinberg:03:burst} that has been used, for example, in discovering trends in citation
literature~\citep{chen:06:citespace}, analyzing topics~\citep{mane:04:topic},
recommending citations~\citep{he:11:citation}, analyzing
disasters~\citep{fontugne:11:disaster}, and analyzing social
networks~\citep{backstrom:06:group} and blogs~\citep{kumar:03:bursty}.

\citet{kleinberg:03:burst} discovers bursts by modelling the time between
events with an exponential model with varying rate parameter. The rate
starts at the base level $\beta$ and can be
raised (multiple times) by a change parameter $\alpha$, but it cannot descend $\beta$. Every time
we raise the parameter, we need to pay a penalty. In the original approach,
the change rate $\alpha$ is given as a parameter and the
base rate is selected to be $\beta = 1/\mu$, where $\mu$ is the average of the
sequence.

We argue that this choice of $\beta$ is suboptimal: (\emph{i}) it does not
maximize the likelihood of the model, and, more importantly, (\emph{ii}) a more
optimized $\beta$ may reveal bursts that would have gone undetected.

We propose a variant of the original burstiness problem, where we are no longer
given the base parameter $\beta$ but instead we are asked to optimize it along
with discovering bursts. We also consider variants where we optimize $\alpha$ as well.
These tweaks are rather mundane from the problem definition
point of view but it leads to a surprisingly difficult optimization problem.

We consider two different models for the delays: exponential and geometric.
First, we will show that we can solve our problem for exponential model in
polynomial time, when $\alpha$ is given. Unfortunately, this algorithm requires $\bigO{n^3k^4}$ time,\!\footnote{
Here, $n$ is the sequence length and $k$ is the maximum number of times the rate can be increased.}
thus being impractical. Even worse, we cannot apply the same approach for
geometric model.  This is a stark contrast to the original approach, where the
computational complexity is $\bigO{nk}$. 

Fortunately, we can estimate burst discovery in quasi-linear
time w.r.t.\ the sequence length; see Table~\ref{tab:algos} for a summary of
the algorithms. We obtain $(1 + \epsilon)$ approximation guarantee for the geometric
model. We also obtain, under some mild conditions, $(1 + \epsilon)$ approximation guarantee
for the exponential model.

In all four cases, the algorithm is simple: we test
multiple values of $\beta$ (and $\alpha$), and use the same efficient dynamic program that is
used to solve the original problem. Among the tested sequences we select the best one.
The main technical challenge is to test the multiple values of $\alpha$ and $\beta$
such that we obtain the needed guarantee while still maintaining a quasi-linear running time with respect to sequence length.

The remainder of the paper is as follows. We review the original burstiness
problem in Section~\ref{sec:prel}, and define our variant in
Section~\ref{sec:problem}. We introduce the exact algorithm in
Section~\ref{sec:exact}, and present the approximation algorithms in
Section~\ref{sec:geoburst}--\ref{sec:expburst}. In Section~\ref{sec:related}, we
present the related work.  In Section~\ref{sec:exp}, we compare demonstrate
empirically that our approach discovers busts that may go unnoticed.  We
conclude with discussion in Section~\ref{sec:conclusions}.
The proofs are given in Appendix, available in the full version of this paper.

\begin{table*}
\caption{Summary of algorithms discussed in this paper.
Here $k$ is the number of allowed levels,
$n$ is the length of the sequence, $\mu$ is the arithmetic mean, and $g$ is the geometric mean,
$\Omega$ is the maximum of the sequence, and
$\omega$ is the minimum of the sequence.  We assume that $\omega > 0$.
$\brstexpprb(\alpha, \beta)$ is the original problem considered by~\citet{kleinberg:03:burst},
and $\brstgeoprb(\alpha, \beta)$ is a minor variation of the problem.
The remaining results are the main contribution of this paper.
}
\begin{tabular*}{\textwidth}{@{\extracolsep{\fill}}lrr}
\toprule
Problem & guarantee & running time \\
\midrule
$\brstexpprb(\alpha, \beta)$ & exact & $\bigO{nk}$ \\
$\brstexpprb(\alpha)$ & exact & $\bigO{n^3k^4}$ \\
$\brstexpprb(\alpha)$ & $\mathit{SOL} - n \log g \leq (1 + \epsilon)(\mathit{OPT} - n\log g)$ & $\bigO{\epsilon^{-1}nk^2 \log \alpha}$ \\
$\brstexpprb$ & $\mathit{SOL} - n \log g \leq (1 + \epsilon)(\mathit{OPT} - n\log g)$ & $\bigO{\epsilon^{-2}nk^3 \log^2 (\Omega / \omega)}$ \\[2mm]
$\brstgeoprb(\alpha, \beta)$ & exact & $\bigO{nk}$ \\
$\brstgeoprb(\alpha)$ & $\mathit{SOL} \leq (1 + \epsilon)\mathit{OPT}$ & $\bigO{\epsilon^{-1} nk \log \log n}$ \\
$\brstgeoprb$ & $\mathit{SOL} \leq (1 + \epsilon)\mathit{OPT}$ & $\bigO{\epsilon^{-2} nk \log (n  \mu  k / \epsilon) \log \log n}$ \\
\bottomrule
\end{tabular*}
\label{tab:algos}
\end{table*} 

\section{Preliminaries}\label{sec:prel}
In this section, we review the setting proposed by~\citet{kleinberg:03:burst},
as well as the dynamic program used to solve this setting.

Assume that we observe an event at different time points, say $t_0, \ldots,
t_n$. The main idea behind discovering bursts is to model the delays between
the events, $s_i = t_i - t_{i - 1}$: if the events occur at higher pace, then
we expect $s_i$ to be relatively small.

Assume that we are given a sequence of delays $S = s_1, \ldots, s_n$.
In order to measure the burstiness of the sequence,
we will model it with an exponential distribution,
	$\pexp{s ; \lambda} = \lambda \exp\fpr{-\lambda s}$.
Larger $\lambda$ dictates that the delays should be shorter, that is, the
events should occur at faster pace.

The idea behind modelling burstiness is to allow the parameter $\lambda$
fluctuate to a certain degree: We start with $\lambda = \beta$, where $\beta$
is a parameter. At any point we can increase the parameter by multiplying with
another parameter $\alpha$. We can also decrease the parameter by dividing
by $\alpha$. We can have multiple increases and decreases, however, we cannot
decrease the parameter below $\beta$. Every time we change the rate from $x$ to $y$, we have
to pay a penalty, $\pen{x, y; \gamma}$, controlled by a parameter $\gamma$.

More formally,
assume that we have assigned the burstiness levels for each
delays $L = \ell_1, \ldots, \ell_n$, where each $\ell_i$ is a non-negative integer.
We will refer to this sequence as the \emph{level sequence}.
For convenience, let us write $\ell_0 = 0$. 
Then the score of burstiness $\scoreexp{L, S; \alpha, \beta, \gamma}$ is equal to
\[
	\sum_{i = 1}^n -\log \pexp{s_i ; \beta\alpha^{\ell_i}} + \pen{\ell_{i - 1}, \ell_i; \gamma}\quad.
\]
The first term---negative log-likelihood of the data---measures how well the
burstiness model fits the sequence, while the second term penalizes the erratic
behavior in $L$. Ideally, we wish to have both terms as small as possible.
To reduce clutter we will often ignore $\gamma$ in notation, as
this parameter is given, and is kept constant.

We will use the penalty function given in~\citep{kleinberg:03:burst},
\[
	\pen{x, y} = \max(y - x, 0) \gamma \log n,
\]
where $n$ is the length of the input sequence. Note that $\pen{}$ depends on $\gamma$ and $n$
but we have suppressed this from the notation to avoid clutter.

We can now state the burstiness problem. 
\begin{problem}[$\brstexpprb(\alpha, \beta)$]
\label{prb:burstexporig}
Given a delay sequence $S$, parameters $\alpha$, $\beta$, $\gamma$, and a
maximum number of levels $k$, find a level sequence $L = \ell_1, \ldots,
\ell_n$, where $\ell_i$ is an integer $0 \leq \ell_i \leq k$,
minimizing $\scoreexp{L, S; \alpha, \beta, \gamma}$.
\end{problem}

Two remarks are in order: First of all, the original problem definition given
by~\citet{kleinberg:03:burst} does not directly use $k$, instead the levels are only limited
implicitly due to $\pen{}$. However, in practice, $k$ is needed by the dynamic
program, but it is possible to select a large enough $k$ such that enforcing
$k$ does not change the optimal sequence~\citep{kleinberg:03:burst}. Since our complexity
analysis will use $k$, we made this constraint explicit. Secondly, the parameter $\beta$
is typically set to $1 / \mu$, where $\mu = \frac{1}{n} \sum s_i$ is the average delay.

We also study an altenative objective.
Exponential distribution is meant primarily for real-valued delays. If the delays
are integers, then the natural counterpart of the distribution is the geometric distribution 
	$\pgeo{s ; \lambda} = (1 - \lambda)\lambda^s$.
Here, \emph{low} values of $\lambda$ dictate that the delays should occur faster.
We can now define 
\[
	\scoregeo{L, S; \alpha, \beta, \gamma} = \sum_{i = 1}^n -\log p_{\mathit{geo}}(s_i ; \beta\alpha^{\ell_i}) + \pen{\ell_{i - 1}, \ell_i}\,.
\]
Note that in $\scoreexp{}$ we use $\alpha > 1$ while here we use $\alpha < 1$.
We can now define a similar optimization problem.
\begin{problem}[$\brstgeoprb(\alpha, \beta)$]
\label{prb:burstgeoorig}
Given an integer delay sequence $S$, parameters $\alpha$, $\beta$, $\gamma$, and a
maximum number of levels $k$, find a level sequence $L = \ell_1, \ldots,
\ell_n$, where $\ell_i$ is an integer $0 \leq \ell_i \leq k$,
minimizing $\scoregeo{L, S; \alpha, \beta, \gamma}$.
\end{problem}

We can solve Problem~\ref{prb:burstexporig} or Problem~\ref{prb:burstgeoorig}
using the standard dynamic programming algorithm by~\citet{viterbi:67:dp}.
Off-the-shelf version of this algorithm requires $\bigO{nk^2}$ time. However, we
can easily speed-up the algorithm to $\bigO{nk}$; for completeness we present this speed-up in Appendix~\ref{sec:app_viterbi}.

\section{Problem definition}\label{sec:problem}

We are now ready to state our problem. 
The difference between our setting and Problem~\ref{prb:burstexporig} is that
here we are asked to optimize $\beta$, and possibly $\alpha$, along with the levels, while in the
original setting $\beta$ was given as a parameter.

We consider two problem variants. In the first variant, we optimize $\beta$
while we are given $\alpha$.

\begin{problem}[$\brstexpprb(\alpha)$]
Given a delay sequence $S$, parameters $\alpha$, $\gamma$, and a
maximum number of levels $k$, find a level sequence $L = \ell_1, \ldots,
\ell_n$, where $\ell_i$ is an integer $0 \leq \ell_i \leq k$,
\emph{and a parameter $\beta$},
minimizing $\scoreexp{L, S; \alpha, \beta, \gamma}$.
\end{problem}

In the second variant, we optimize both $\alpha$ and $\beta$. 

\begin{problem}[$\brstexpprb$]
Given a delay sequence $S$, a parameter $\gamma$, and a
maximum number of levels $k$, find a level sequence $L = \ell_1, \ldots,
\ell_n$, where $\ell_i$ is an integer $0 \leq \ell_i \leq k$,
\emph{and parameters $\alpha$ and $\beta$},
minimizing $\scoreexp{L, S; \alpha, \beta, \gamma}$.
\end{problem}

While this modification is trivial and mundane from the problem definition
point of view, it carries several crucial consequences. First of all,
optimizing $\beta$ may discover bursts that would otherwise be undetected.
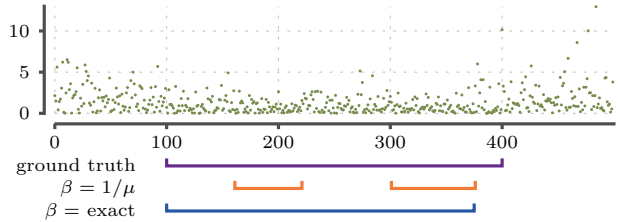
\begin{figure}[ht!]
\begin{center}
\begin{tikzpicture}[baseline]
\begin{axis}[
    height = 3cm,
    width = 9cm,
    cycle list name=yaf,
	clip mode = individual,
    ]

\addplot[only marks, mark size = 0.3, mark = *, yafcolor3, mark options = {line width = 0.5pt}]
	table[x expr = \coordindex, y index = 0, header = false]  {meanvsexact.dat};

\coordinate (b1) at (0, -0.7cm);
\coordinate (b2) at (0, -1cm);
\coordinate (b3) at (0, -1.3cm);

\node[anchor = east, font = \scriptsize] at (1.2cm, -0.7cm) {ground truth};
\node[anchor = east, font = \scriptsize] at (1.2cm, -1cm) {$\beta = 1 / \mu$};
\node[anchor = east, font = \scriptsize] at (1.2cm, -1.3cm) {$\beta =  $ exact};

\coordinate (s) at (axis cs:100, 0);
\coordinate (e) at (axis cs:400, 0);
\drawinterval{b1-|s}{b1-|e}{interval};

\coordinate (s) at (axis cs:161, 0);
\coordinate (e) at (axis cs:221, 0);
\drawinterval{b2-|s}{b2-|e}{interval, yafcolor2};
\coordinate (s) at (axis cs:301, 0);
\coordinate (e) at (axis cs:376, 0);
\drawinterval{b2-|s}{b2-|e}{interval, yafcolor2};

\coordinate (s) at (axis cs:100, 0);
\coordinate (e) at (axis cs:375, 0);
\drawinterval{b3-|s}{b3-|e}{interval, yafcolor5};

\pgfplotsextra{\yafdrawaxis{1}{500}{0}{13}}
\end{axis}
\end{tikzpicture}
\end{center}
\caption{A toy data set $S$ with a burst between 100 and 400. Low values indicate
short delays, bursts. The indicated regions are
(\emph{i}) the ground truth,
(\emph{ii}) bursts discovered with $\beta = 1 / \mu$, where $\mu$ is the average delay, and
(\emph{iii}) bursts discovered with $\beta$ set to the exact value of the generative model.
}
\label{fig:meanvsexact}
\end{figure}

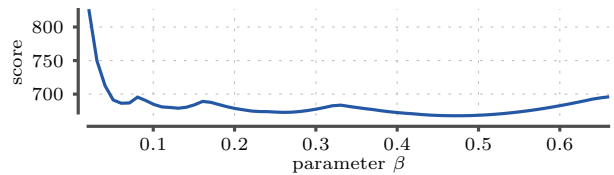
\begin{figure}[ht!]
\begin{center}
\begin{tikzpicture}[baseline]
\begin{axis}[xlabel={parameter $\beta$}, ylabel= {score},
    height = 3cm,
    width = 8.5cm,
    cycle list name=yaf,
	clip mode = individual,
	xtick = {0.02, 0.1, 0.2, 0.3, 0.4, 0.5, 0.6}
    ]

\addplot[yafcolor5] table[x index = 0, y index = 1, header = false]  {nonconvex.dat};

\pgfplotsextra{\yafdrawaxis{0.02}{0.66}{672}{826}}
\end{axis}
\end{tikzpicture}
\end{center}
\caption{Score $\pexp{S, L^*; \alpha, \beta, k}$ as a function of $\beta$,
where $\alpha = 2$, $k = 4$, and $L^*$ is the optimal level sequence for the given parameters.
Low
values are better. 
}
\label{fig:nonconvex}
\end{figure}

\begin{example}
Consider a sequence given in Figure~\ref{fig:meanvsexact}, which shows a
sequence of $500$ delays. The burst between 100 and 400 is generated using
exponential model with $\lambda = 1/2$, the remaining delays are generated
using $\lambda = 1$. We applied \algdp with $\beta^{-1}$ equal to the average
of the sequence, the value used by~\citet{kleinberg:03:burst}, and compare it to $\beta = 1/2$,
which is the correct ground level of the generative model.
The remaining parameters were set to $\alpha = 2$, $\gamma = 1$, and $k = 1$.  We see that in the
latter case we discover a burst that is much closer to the ground truth.\qed

\end{example}

Our second remark is that if $\alpha$ and $\beta$ are given, one can easily discover the
optimal bursts using \algdp. The optimization becomes non-trivial when we need to optimize $\alpha$ and $\beta$
as well. To make matters worse, the score as a function of $\beta$ is non-convex, as
demonstrated in Figure~\ref{fig:nonconvex}. Hence, we can easily get stuck
in local minima.

Next,
we introduce discrete variants of $\brstexpprb(\alpha)$ and \brstexpprb.

\begin{problem}[$\brstgeoprb(\alpha)$]
Given an integer delay sequence $S$, parameters $\alpha$, $\gamma$, and a
maximum number of levels $k$, find a level sequence $L = \ell_1, \ldots,
\ell_n$, where $\ell_i$ is an integer $0 \leq \ell_i \leq k$,
\emph{and a parameter $\beta$},
minimizing $\scoregeo{L, S; \alpha, \beta, \gamma}$.
\end{problem}

\begin{problem}[$\brstgeoprb$]
Given an integer delay sequence $S$, a parameter $\alpha$, and a
maximum number of levels $k$, find a level sequence $L = \ell_1, \ldots,
\ell_n$, where $\ell_i$ is an integer $0 \leq \ell_i \leq k$,
\emph{and parameters $\alpha$ and $\beta$},
minimizing $\scoregeo{L, S; \alpha, \beta, \gamma}$.
\end{problem}

Despite being very similar problems, we need to analyze these problems
individually.  We will show that $\brstexpprb(\alpha)$ can be solved exactly in
polynomial time, although, the algorithm is too slow for practice.
This approach does not work for other problems but we will show
that all four problems can be $(1 + \epsilon)$-approximated efficiently.

Before we continue, we need to address a pathological case when solving
\brstexpprb: the problem of \brstexpprb is illdefined if the delay sequence
$S$ contains a zero. To see this, assume that $s_i = 0$. Then a level sequence
$\ell_i = 1$, and $\ell_j = 0$, for $j \neq i$, with $\alpha = \infty$ and
$\beta = 1$ leads to a score of $-\infty$. This is because $\pexp{s_i; \beta\alpha} = \infty$
and the remaining terms are finite. This is why we assume that whenever we deal with
\brstexpprb, we have $s_i > 0$.
If we have $s_i = 0$, then we can either set $\alpha$ manually by using $\brstexpprb(\alpha)$
or shift the delays by a small amount.

\section{Exact algorithm for $\brstexpprb(\alpha)$}\label{sec:exact}
In this section we present an exact polynomial algorithm for solving $\brstexpprb(\alpha)$.
Unfortunately, this algorithm is impractically slow for large sequences: the time complexity is $\bigO{n^3k^4}$ and
the space complexity is $\bigO{n^3k^3}$. Thus, it only serves as a theoretical result.
More practical algorithms are given in the next sections.

In order to solve \brstexpprb we introduce a more complicated optimization problem.

\begin{problem}[\boundprb]
Given a delay sequence $S = s_1, \ldots, s_n$, a parameter $\alpha$, budget parameters $d$ and $m$,
and a maximum number of levels $k$, find a level sequence $L = \ell_1, \ldots, \ell_n$,
with $0 \leq l_i \leq k$, minimizing 
\[
\begin{split}
	&\sum_{i = 1}^n \alpha^{l_i} s_i\quad \text{such that} \\
	&\sum_i \max\pr{\ell_i - \ell_{i - 1}, 0} = d \quad\text{and}\quad
	\sum_i \ell_i = m\quad. 
\end{split}
\]
\end{problem}

We will show that this problem can be solved in polynomial time. But before, let
us first show that \brstexpprb and \boundprb are intimately connected. See Appendix~\ref{sec:app_boundprb} for the proof.

\begin{proposition}
\label{prop:boundprb}
Assume a delay sequence $S$, and parameters $\alpha$ and $\gamma$, and an upper
bound for levels $k$. 
There are budget parameters $d \leq k(n + 1)/2$ and $m \leq kn$ for which
the level sequence $L$ solving \boundprb also solves
$\brstexpprb(\alpha)$ along with
\[
	\beta = \frac{n}{\sum_i s_i \alpha^{\ell_i}}\quad.
\]
\end{proposition}

We can solve \boundprb with a dynamic program. In order to do this, let us
define a table $o$, where an entry $o[i, j, a, b]$ is the optimal score of the first $i$ symbols
of the input sequence such that 
\[
	\ell_i = j,\ 
	\sum_{x = 1}^i \max\pr{\ell_x - \ell_{x - 1}, 0} = a, \ \text{and}\ 
	\sum_{x = 1}^i \ell_x = b\quad.
\]
In case, there is no level sequence satisfying the constraints, we set $o[i, j, a, b] = \infty$. 
Due to Proposition~\ref{prop:boundprb}, we can limit $a \leq k(n + 1)/2$ and $b \leq kn$.
Consequently, $o$ contains $\bigO{n^3k^3}$ entries.
We can compute a single entry with 
\begin{equation}
\begin{split}
	 &o[i, j, a, b]  \\
	 &\quad =  \alpha^j s_i  + \min_{j'} o[i - 1, j', a - \max(0, j - j'), b - j]\quad.
\label{eq:exactdp}
\end{split}
\end{equation}
The computation of a single value thus requires $\bigO{k}$ time. So computing the whole table
can be done in $\bigO{n^3k^4}$. Moreover, if we also store the optimal $j'$
as given in Equation~\ref{eq:exactdp}, for each cell, we can recover the level
sequence responsible for every $o[i, j, a, b]$. 

Proposition~\ref{prop:boundprb} now guarantees that we can solve \brstexpprb by
comparing the level sequences responsible for $o[n, j, a, b]$, where $j = 0, \ldots, k$, $a = 0, \ldots, (k + 1)n/2$,
and $b = 0, \ldots, kn$.

\section{Approximating discrete burstiness}\label{sec:geoburst}

In this section we will provide a $(1 + \epsilon)$-approximation algorithms for
$\brstgeoprb(\alpha)$ and $\brstgeoprb$. The time complexities are stated in Table~\ref{tab:algos}.

\subsection{Approximating $\brstgeoprb(\alpha)$}

Note that if we knew the optimal $\beta$, then $\brstgeoprb(\alpha)$ reduces to Problem~\ref{prb:burstgeoorig},
which we can solve in $\bigO{nk}$ time by applying \algdp.
The idea behind our approximation is 
to test several values of $\beta$, and select the best solution among the tested values.
The trick is to select values densely enough so that we can obtain $(1 + \epsilon)$ guarantee
while keeping the number of tests low, namely $\bigO{\epsilon^{-1}\log \log n}$.
The pseudo-code of the algorithm is given in Algorithm~\ref{alg:geoapproxalpha}.

\begin{algorithm}[ht]
	$\mu \define  \frac{1}{n}\sum_{i} s_i$\;
	\lIf {$\mu = 0$} {
		\Return $L = \enpr{0}{0}$ 
	}
	$\eta \define \mu /(\mu + 1)$\;
	$c \define 1$\;
	\While {$\eta^c \leq  \mu/(\mu + 1/n)$} {
		$\beta \define \eta^c$\;
		$L \define \algdp(S, \alpha, \beta, \gamma, k, \pgeo{})$\;
		$c \define c / (1 + \epsilon)$\;
	}
	\Return the best observed $L$\;
\caption{$\alggeoalpha(S, \alpha, \gamma, k, \epsilon)$}
\label{alg:geoapproxalpha}
\end{algorithm}

Next we state that the algorithm indeed yields an
$(1 + \epsilon)$-approximation ratio, and can be executed in
$\bigO{\epsilon^{-1}nk\log \log n}$ time. 
The proofs are given in Appendix~\ref{sec:app_geoapproxalpha}--\ref{sec:app_geotimealpha}. 

\begin{proposition}
\label{prop:geoapproxalpha}
Let $S$ be an integer delay sequence, and let $\alpha$, $\gamma$, and $k$ be the parameters.
Let $L^*$, $\beta^*$ be the solution to $\brstgeoprb(\alpha)$.  
Assume $\epsilon > 0$.
Let $L$, $\beta$ be the solution returned by $\alggeoalpha(S, \alpha, \gamma, k, \epsilon)$.  
Then 
\[
	\scoregeo{S, L; \beta} \leq (1 + \epsilon)\scoregeo{S, L^*; \beta^*}\quad.
\]
\end{proposition}

\begin{proposition}
\label{prop:geotimealpha}
The computational complexity of
\alggeoalpha is $\bigO{\epsilon^{-1}nk \log \log n}$.
\end{proposition}

\subsection{Approximating $\brstgeoprb$}

We now turn to approximating \brstgeoprb. The approach here is
similar to the previous approach: we test multiple values
of $\alpha$ and invoke \alggeoalpha. The pseudo-code for the algorithm
is given in Algorithm~\ref{alg:geoapprox}.

\begin{algorithm}[ht]
	$L \define \alggeoalpha(S, 0, \gamma, k, \epsilon)$\;
	$\mu \define  \frac{1}{n}\sum_{i} s_i$\;
	$\eta \define 1 /(1 + nk)$\;
	$\sigma \define \mu/(\mu + 1/n)$\;
	$c \define 1$\;
	\While {$\eta^c \leq  \sigma^{\epsilon / k}$} {
		$\alpha \define \eta^c$\;
		$L \define \alggeoalpha(S, \alpha, \gamma, k, \epsilon)$\;
		$c \define c / (1 + \epsilon)$\;
	}
	\Return the best observed $L$\;
\caption{$\alggeo(S, \gamma, k, \epsilon)$}
\label{alg:geoapprox}
\end{algorithm}

Next we establish the correctness of the method as well as the running time.
The proofs are given in Appendix~\ref{sec:app_geoapprox}--\ref{sec:app_geotime}. 

\begin{proposition}
\label{prop:geoapprox}
Let $L^*$, $\alpha^*$, $\beta^*$ be the solution to \brstgeoprb.
Assume $\epsilon > 0$.
Let $L$, $\alpha$, $\beta$ be the solution returned by $\alggeo(S, \gamma, k, \epsilon)$.
Then
\[
    \scoregeo{S, L; \alpha, \beta} \leq (1 + \epsilon)\scoregeo{S, L^*; \alpha^*, \beta^*}\quad.
\]
\end{proposition}

\begin{proposition}
\label{prop:geotime}
The computational complexity of
\alggeo is 
\[
\bigO{nk \log \log n(\log n + \log \mu + \log k - \log\epsilon)\epsilon^{-2}}\quad.
\]
\end{proposition}

\section{Approximating continuous burstiness}\label{sec:expburst}

In this section we will provide a $(1 + \epsilon)$-approximation algorithms for
$\brstexpprb(\alpha)$ and $\brstexpprb$. The time complexities are stated in Table~\ref{tab:algos}.

\subsection{Approximating $\brstexpprb(\alpha)$}
In this section we introduce an approximation algorithm for $\brstexpprb(\alpha)$.  The
general approach of this algorithm is the same as in \alggeoalpha: we test several
values of $\beta$, solve the resulting subproblem with \algdp, and select the
best one. The pseudo-code is given in
Algorithm~\ref{alg:expalpha}.

\begin{algorithm}[ht]
	$\mu \define  \frac{1}{n}\sum_{i} s_i$\;
	$\beta \define 1 / \mu$\;
	\While {$\beta \geq  1/(\alpha^k\mu)$} {
		$L \define \algdp(S, \alpha, \beta, \gamma, k, \pexp{})$\;
		$\beta \define \beta / (1 + \epsilon)$\;
	}
	\Return the best observed $L$ and $\beta$\;
\caption{$\algexpalpha(S, \alpha, \gamma, k, \epsilon)$}
\label{alg:expalpha}
\end{algorithm}

Unlike with \alggeoalpha, \algexpalpha does not yield an unconditional $(1 +
\epsilon)$-approximation guarantee. The key problem is that since exponential
distribution is continuous, the term $\pexp{s; \lambda}$ may be larger than
$1$.  Consequently, $-\log \pexp{s; \lambda}$, as well as the actual score
$\scoreexp{}$, can be negative.  However, if the delay sequence has a geometric
mean larger or equal than 1, we can guarantee the approximation ratio.

The proofs for the next two propositions are given in
Appendix~\ref{sec:app_expapproxalpha}--\ref{sec:app_exptimealpha}.

\begin{proposition}
\label{prop:expapproxalpha}
Assume a delay sequence $S$, and parameters $\alpha$ and $\gamma$,
and an upper bound for levels $k$.
Let $\beta^*$ and $L^*$ be the solution to $\brstexpprb(\alpha(\alpha))$.
Let $g = \spr{\prod_{i} s_i}^{1/n}$ be the geometric mean. 
Assume $\epsilon > 0$.
Let $L$, $\beta$ be the solution returned by \algexpalpha.
Then
\[
	\scoreexp{S, L; \beta} - n \log g \leq  (1 + \epsilon)(\scoreexp{S, L^*; \beta^*} - n \log g)\quad.
\]
Moreover, if $g \geq 1$, then
\[
	\scoreexp{S, L; \beta} \leq  (1 + \epsilon)\scoreexp{S, L^*; \beta^*} \quad.
\]
\end{proposition}

Note that if the geometric mean $g$ is less than 1, then we still have a
guarantee, except now we need to shift the score by a (positive) constant of $-n\log g$.

\begin{proposition}
\label{prop:exptimealpha}
The computational complexity of
\algexpalpha is $\bigO{\epsilon^{-1}nk^2 \log \alpha}$.
\end{proposition}

\subsection{Approximating $\brstexpprb$}

We now turn to approximating \brstexpprb. The approach here is
similar to the previous approach: we test multiple values
of $\alpha$ and invoke \algexpalpha. The pseudo-code for the algorithm
is given in Algorithm~\ref{alg:expapprox}.

\begin{algorithm}[ht]
	$\alpha \define (\max s_i) / (\min s_i)$\;
	$c \define \sqrt[2k]{1 + \epsilon}$\;
	\While {$\alpha \geq  1$} {
		$L \define \algexpalpha(S, \alpha, \gamma, k, \epsilon / 2)$\;
		$\alpha \define \alpha / c$\;
	}
	\Return the best observed $L$\;
\caption{$\algexp(S, \gamma, k, \epsilon)$}
\label{alg:expapprox}
\end{algorithm}

Next we establish the correctness of the method as well as the running time.
The proofs given in
Appendix~\ref{sec:app_expapprox}--\ref{sec:app_exptime}.

\begin{proposition}
\label{prop:expapprox}
Assume a delay sequence $S$, a parameter $\gamma$,
and an upper bound for levels $k$.
Let $\alpha^*$, $\beta^*$ and $L^*$ be the solution to \brstexpprb.
Let $g = \spr{\prod_{i} s_i}^{1/n}$ be the geometric mean, 
and let $\psi = n \log g$.
Assume $\epsilon > 0$.
Let $L$, $\alpha$, $\beta$ the solution returned by \algexp.
Then
\[
	\scoreexp{S, L; \alpha, \beta} - \psi \leq  (1 + \epsilon)(\scoreexp{S, L^*; \alpha^*, \beta^*} - \psi)\quad.
\]
Moreover, if $g \geq 1$, then
\[
	\scoreexp{S, L; \alpha, \beta} \leq  (1 + \epsilon)\scoreexp{S, L^*; \alpha, \beta^*} \quad.
\]
\end{proposition}

\begin{proposition}
\label{prop:exptime}
Let $\Omega = \max s_i$ and let $\omega = \min s_i$.
The computational complexity of
\algexp is $\bigO{\epsilon^{-2}nk^3 \log^2(\Omega / \omega)}$.
\end{proposition}

\subsection{Speeding up $\brstexpprb(\alpha)$}
Our final step is to describe how can we speed-up the computation of $\brstexpprb(\alpha)$ in practice.
The following proposition allows us to ignore a significant amount of tests.

\begin{proposition}
\label{prop:speedup}
Assume a delay sequence $S$, and parameters $\alpha$ and $\gamma$.
Let $\beta$ be a parameter, and let $L$ be the optimal solution for $\brstexpprb(\alpha, \beta)$.
Define
\[
	\beta' = \frac{n}{\sum_i s_i \alpha^{\ell_i}}\quad.
\]
Let $\beta^*$ be the optimal parameter to $\brstexpprb(\alpha)$. Then
either
\[
\beta^* \leq \min(\beta, \beta') \quad\text{or}\quad
\beta^* \geq \max(\beta, \beta')\quad.
\]
\end{proposition}

Proposition~\ref{prop:speedup} allows us to ignore some tests: Let $\beta_i$ be the parameters
tested by \algexpalpha, that is, $\beta_i =
\mu^{-1}(1 + \epsilon)^{-i}$. 
Assume that we test $\beta_i$, and compute $\beta'$ as given in Proposition~\ref{prop:speedup}.
If $\beta' > \beta_i$, we can safely ignore testing any $\beta_j$ such that $\beta_i < \beta_j < \beta'$.
Similarly,
if $\beta' < \beta_i$, we can safely ignore testing any $\beta_j$ such that $\beta' < \beta_j < \beta_i$.

The testing order of $\beta_i$ matters since we want to use
both cases $\beta' < \beta_i$ and $\beta' > \beta_i$  efficiently. We propose the following order which worked
well in our experimental evaluation: Let $t$ be the number of different $\beta_i$, and let $m$ be the largest integer
for which $2^m \leq t$. Test the parameters in the order
\[
	0, 2^m, 2^{m - 1}, 2^{3(m - 1)}, \ldots, 1, 3, 5, 7, \ldots,
\]
that is, we start with $0$ and increment by $2^m$ until we reach the end of the list. Then we decrease $m$ by 1, and repeat.
During the traverse, we ignore the parameters that were already tested, as well as the redundant parameters. 

Interestingly enough, this approach cannot be applied directly to the discrete
version of the problem.  First of all, the technique for proving
Proposition~\ref{prop:speedup} cannot be applied directly to the score function
for the geometric distribution. Secondly, there is no closed formula for
computing the discrete analogue of $\beta'$ given in
Proposition~\ref{prop:speedup}. 

\section{Related work}\label{sec:related}

\textbf{Discovering bursts}
Modelling and discovering bursts is a very well-studied topic in data mining.
We will highlight some existing techniques.
We are modelling delays between events, but we can alternatively model
event counts in some predetermined window: high count indicate burst.
\citet{ihler:06:adaptive} proposed modelling such a statistic with Poisson
process, while \citet{fung:05:burst} used Binomial distribution.
If the events at hand are documents, we can model burstiness
with time-sensitive topic models~\citep{wang:06:tot,leskovec:06:topic,kawamae:11:trend}.
As an alternative methods to discover bursts, \citet{zhu:03:wavelet} used
wavelet analysis, \citet{vlachos:04:identify} applied Fourier analysis, and
\citet{he:10:topic} adopted concepts from Mechanics.
\citet{lappas:09:burst} propose discovering maximal bursts with large discrepancy.

\textbf{Segmentation}
A sister problem of burstiness is a classic segmentation problem. Here instead
of penalizing transitions, we limit the number of segments to
$k$. If the overall score is additive w.r.t. the segments, then this problem
can be solved in $\bigO{n^2k}$ time~\citep{bellman:61:on}. For certain cases, this problem 
has a linear time solution~\citep{galil:90:linear}. Moreover, under some mild assumptions we
can obtain a $(1 + \epsilon)$ approximation in linear time~\citep{guha:06:estimate}.

\textbf{Concept drift detection in data streams:} A related problem setting to
burstiness is concept drift detection. Here, a typical goal is to have an
online algorithm that can perform update quickly and preferably does not use
significant amount of memory. For an overview of existing techniques see an
excellent survey by~\citet{gama:14:survey}.  The algorithms introduced in this
paper along with the original approach are not strictly online because in every
case we need to know the mean of the sequence. However, if the mean is known,
then we can run \algdp in online fashion, and,
if we are only interested in the burstiness of a current symbol,
we need to maintain only $\bigO{k}$ elements, per $\beta$.

\section{Experimental evaluation}\label{sec:exp}

In this section we present our experiments.
As a baseline we use method by~\citet{kleinberg:03:burst}, that is, we
derive the parameter $\beta$ from $\mu$, the mean of the sequence.
For exponential model, $\beta = \mu^{-1}$; we refer to this model as
\algmeanexp. For geometrical model, $\beta = \mu / (\mu + 1)$; we refer to this
approach as \algmeangeo. Throughout the experiments, we used $\epsilon = 0.05$
and $\gamma = 1$ for our algorithms. 

\textbf{Experiments with synthetic data:}
We first focus on demonstrating when optimizing $\beta$ is more
advantageous than the baseline approach.

For our first experiment we generated a sequence of $500$ data points.
We planted a single burst with a varying length $50$--$250$. The burst was
generated with $\pexp{\cdot; 1}$, while the remaining sequence was generated
with $\pexp{\cdot; 2}$.  We computed bursts with \algmeanexp and \algexpalpha,
the parameters were set to $k = 1$, $\alpha = 2$. The obtained level sequence
was evaluated by computing the hamming distance, $\sum_i \abs{\ell_i - \ell^*_i}$, where $\ell^*_i$ 
is the ground truth level sequence.
We repeated each experiment $100$ times.

We see from the results given in Figure~\ref{fig:burstlen} that the bursts
discovered by \algexpalpha are closer to the ground truth, on average, than the
baseline.  This is especially the case when burst becomes larger. The main
reason for this is that short bursts do not affect significantly the average of
the sequence, $\mu$, so consequently, $\mu$ is close to the base activity
level.  As the burst increases, so does $\mu$, which leads to 
underestimating of $\beta$.

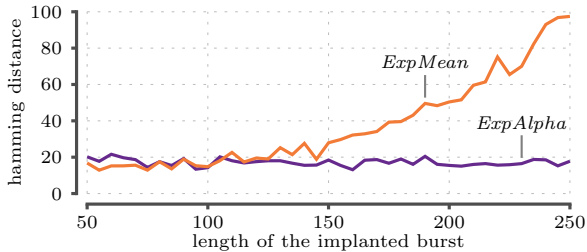
\begin{figure}[ht!]
\begin{center}
\begin{tikzpicture}[baseline]
\begin{axis}[xlabel={length of the implanted burst}, ylabel= {hamming distance},
    height = 4cm,
    width = 8cm,
    cycle list name=yaf,
	ymax = 100,
	ymin = 0,
	xtick = {50, 100, ..., 250}
    ]

\addplot[yafcolor1] table[x index = 0, y index = 1, header = false]  {burstlength2.dat};
\addplot[yafcolor2] table[x index = 0, y index = 2, header = false]  {burstlength2.dat};

\node[inner sep = 2pt, pin = {[pin distance = 3mm, pin edge = {thick}, text = black, font = \scriptsize, inner sep = 1pt]90:\algexpalpha}] at (axis cs: 230, 16.46) {};
\node[inner sep = 2pt, pin = {[pin distance = 3mm, pin edge = {thick}, text = black, font = \scriptsize, inner sep = 1pt]90:\algmeanexp}] at (axis cs: 190, 49.58) {};

\pgfplotsextra{\yafdrawaxis{50}{250}{0}{100}}
\end{axis}
\end{tikzpicture}
\caption{Hamming distance between the ground truth and the discovered level sequence as a function of the length of the planted burst. 
Low values are better.}
\label{fig:burstlen}
\end{center}
\end{figure}

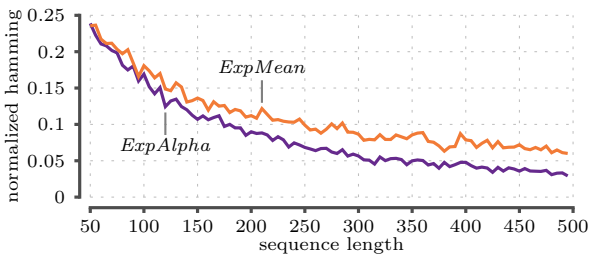
\begin{figure}[ht!]
\begin{center}
\begin{tikzpicture}[baseline]
\begin{axis}[xlabel={sequence length}, ylabel= {normalized hamming},
    height = 4cm,
    width = 8cm,
    cycle list name=yaf,
	xmax = 500,
	ymin = 0,
	ymax = 0.25,
	yticklabel style={/pgf/number format/fixed,/pgf/number format/precision=2},
	ytick = {0, 0.05, ..., 0.30},
	xtick = {50, 100, ..., 500}
    ]

\addplot[yafcolor1] table[x index = 0, y expr = {\thisrowno{2} / \thisrowno{0}}, header = false]  {sequence_length.dat};
\addplot[yafcolor2] table[x index = 0, y expr = {\thisrowno{3} / \thisrowno{0}}, header = false]  {sequence_length.dat};

\node[inner sep = 2pt, pin = {[pin distance = 3mm, pin edge = {thick}, text = black, font = \scriptsize, inner sep = 1pt]270:\algexpalpha}] at (axis cs: 120, 14.93 / 120) {};
\node[inner sep = 2pt, pin = {[pin distance = 3mm, pin edge = {thick}, text = black, font = \scriptsize, inner sep = 1pt]90:\algmeanexp}] at (axis cs: 210, 25.56 / 210 ) {};

\pgfplotsextra{\yafdrawaxis{50}{500}{0}{0.25}}
\end{axis}
\end{tikzpicture}
\end{center}
\caption{Hamming distance, normalized by the sequence length, between the
ground truth level sequence and the discovered level sequence as a function of
the sequence length. Low values are better.}
\label{fig:burstlen2}
\end{figure}

Our next experiment is similar, expect now we vary the sequence length, $n$,
(50--500) and set the burst length to be $n/3$. We generated the sequence as
before, and we use the same parameters.  In
Figure~\ref{fig:burstlen2} we report,
$\frac{1}{n}\sum_i \abs{\ell_i - \ell^*_i}$,
the number of disagreements compared with the ground truth, normalized by
$n$.
Each experiment was repeated 300 times.

We see that for the shortest sequences, the number of disagreement is same for
both algorithm, around $0.2$--$0.25$. This is due that we do not have enough
samples to override the transition penalty $\pen{}$. Once the sequence becomes
longer, we have more evidence of a burst, and here \algexpalpha starts to beat the
baseline, due to a better model fit.

\textbf{Experiments with real-world data:}
We considered two datasets: The first dataset, \dtname{Crimes}, consists of
17\,033 crimes related to narcotics in Chicago between January and October, 2015.
The second dataset, \dtname{Mine}, consists of 909 fatalities in U.S. mining
industry dating from 2000, January.\!\footnote{Both datasets are available at
\url{http://data.gov/}.}
This data is visualized in Figure~\ref{fig:minedata}.
In both datasets, each event has a time stamp: in \dtname{Crimes} we use minutes as granularity, whereas in \dtname{Mine} the time stamp is by the date.
Using these time stamps, we created a delay sequence.

We applied \algexp, \algexpalpha, and \algmeanexp to \dtname{Crimes}. We set $k = 4$,
and for \algexpalpha and \algmeanexp we used $\alpha = 2$.
Since \dtname{Crimes} contains events with 0 delay, we added 1 minute to each delay to avoid the pathological case described in Section~\ref{sec:problem}.
The obtained bursts are presented in Figure~\ref{fig:crimes}.
We also applied \alggeo, \alggeoalpha, and \algmeangeo to \dtname{Mine}. Here we set $\alpha =
1/2$ and $k = 4$, however the algorithm used only 3 levels. The obtained bursts are presented in Figure~\ref{fig:minedelays}.

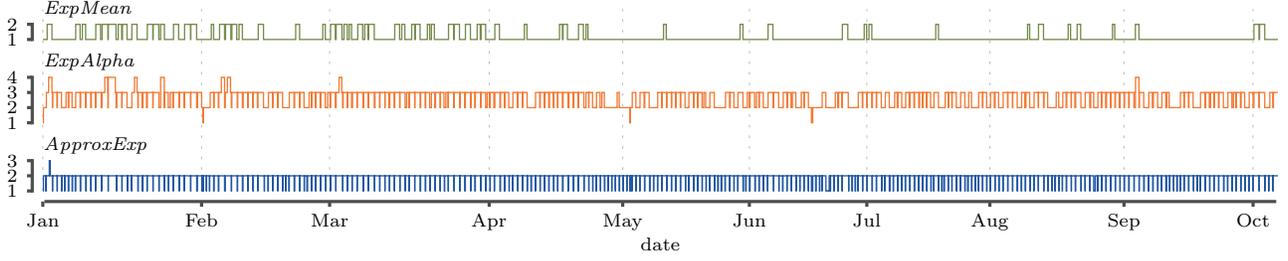
\begin{figure*}
\begin{tikzpicture}[baseline]
\begin{axis}[xlabel={date}, 
    height = 4cm,
    width = 18cm,
    cycle list name=yaf,
	xtick = { -342, -2122, -3972, -5669, -7289, -9036, -10877, -13078, -14844, -17033},
	xticklabels = {Oct, Sep, Aug, Jul, Jun, May, Apr, Mar, Feb, Jan},
	scaled x ticks = false,
	ytick = {0, 1, 2, 4.5, 5.5, 6.5, 7.5, 10, 11},
	xmin = {-17033},
	ymax = 12,
	yticklabels = {1, 2, 3, 1, 2, 3, 4, 1, 2},
	ymajorgrids = false,
	every axis plot post/.append style= {line width=0.5pt} ,
    ]

\addplot[yafcolor5, const plot]
	table[x expr = {-\thisrowno{0}}, y expr = {\thisrowno{1}}, header = false]  {chicago_narcotics_approx_compressed.out};

\addplot[yafcolor2, const plot]
	table[x expr = {-\thisrowno{0}}, y expr = {\thisrowno{1} + 4.5}, header = false]  {chicago_narcotics_alpha_compressed.out};

\addplot[yafcolor3, const plot]
	table[x expr = {-\thisrowno{0}}, y expr = {\thisrowno{1} + 10}, header = false]  {chicago_narcotics_mean_compressed.out};

\node[anchor = west, font = \scriptsize, inner sep = 0pt] at (axis cs: -17033, 3) {\algexp};
\node[anchor = west, font = \scriptsize, inner sep = 0pt] at (axis cs: -17033, 8.5) {\algexpalpha};
\node[anchor = west, font = \scriptsize, inner sep = 0pt] at (axis cs: -17033, 12) {\algmeanexp};

\pgfplotsextra{\yafdrawxaxis{0}{-17033}\yafdrawyaxis{0}{2}\yafdrawyaxis{4.5}{7.5}\yafdrawyaxis{10}{11}}
\end{axis}
\end{tikzpicture}
\caption{Discovered bursts in \dtname{Crimes} dataset.}
\label{fig:crimes}
\end{figure*}

\begin{figure}[ht!]
\begin{tikzpicture}[baseline]
\begin{axis}[xlabel={year (from 2000)}, ylabel= {days w/o fatality},
    height = 3.2cm,
    width = 9cm,
    cycle list name=yaf,
    clip mode = individual,
	xticklabels = {00, 01, 02, 03, 04, 05, 06, 07, 08, 09, 10, 11, 12, 13, 14, 15},
	xtick = {0, 96, 174, 247, 307, 369, 429, 503, 570, 623, 658, 730, 766, 802, 843, 886}
    ]

\addplot[yafcolor5, only marks, mark = *, mark size = 0.3]
table[x expr = {\coordindex}, y index = 0, header = false]  {fatalities_delays.dat};

\pgfplotsextra{\yafdrawaxis{0}{907}{0}{66}}
\end{axis}
\end{tikzpicture}\\
\caption{The delay sequence \dtname{Mine}, as well as the discovered bursts.} 
\label{fig:minedata}
\end{figure}

\begin{figure}[ht!]
\begin{tikzpicture}[baseline]
\begin{axis}[xlabel={year (from 2000)}, 
    height = 3.7cm,
    width = 9.3cm,
    cycle list name=yaf,
    clip mode = individual,
	xticklabels = {00, 01, 02, 03, 04, 05, 06, 07, 08, 09, 10, 11, 12, 13, 14, 15},
	xtick = {0, 96, 174, 247, 307, 369, 429, 503, 570, 623, 658, 730, 766, 802, 843, 886},
	ytick = {0, 1, 3.5, 4.5, 5.5, 8, 9, 10},
	ymax = 11,
	yticklabels = {1, 2, 1, 2, 3, 1, 2, 3},
	ymajorgrids = false,
    ]

\addplot[yafcolor5, const plot]
table[x index = 0, y expr = {\thisrowno{1}}, header = false]  {fatalities_delays_approx_compressed.out};

\addplot[yafcolor2, const plot]
table[x index = 0, y expr = {\thisrowno{1} + 3.5}, header = false]  {fatalities_delays_mean_compressed.out};

\addplot[yafcolor3, const plot]
table[x index = 0, y expr = {\thisrowno{1} + 8}, header = false]  {fatalities_delays_mean_compressed.out};

\node[anchor = west, font = \scriptsize, inner sep = 0pt] at (axis cs: 0, 2) {\algexp};
\node[anchor = west, font = \scriptsize, inner sep = 0pt] at (axis cs: 0, 5.5) {\algexpalpha};
\node[anchor = west, font = \scriptsize, inner sep = 0pt] at (axis cs: 0, 10) {\algmeanexp};

\pgfplotsextra{\yafdrawxaxis{0}{907}\yafdrawyaxis{0}{1}\yafdrawyaxis{3.5}{5.5}\yafdrawyaxis{8}{10}}
\end{axis}
\end{tikzpicture}
\caption{The delay sequence \dtname{Mine}, as well as the discovered bursts.} 
\label{fig:minedelays}
\end{figure}

In \dtname{Mine}, the results by \alggeoalpha and \algmeangeo are the same.
However, we noticed that the results differ if we use different $\alpha$.  The
biggest difference between \alggeo and \alggeoalpha is the last burst:
\alggeoalpha (and \algmeangeo) set the last burst to be on level 2, while
\alggeo uses level 1. The reason for this is that \alggeo selects $\alpha$ to
be very close to $0$, that is, much smaller than $1/2$, the parameter used by
the other algoritms. This implies that when going one level up, the model
expects the events to be much closer to each other. 

In \dtname{Crimes}, \algexp and \algexpalpha discover burstier structure than
\algmeanexp. \algexpalpha uses 4 different levels.  Interestingly enough, in
this level sequence, we spent most of the time at level 1, and we descended to
level 0 for 3 short bursts. In other words, in addition to finding crime
streaks, \algexpalpha also found three short periods when narcotics related
crime rate was lower than usual. \algexp also spends most of its time on level
1 but often descends on level 0, while also highlighting one burst in early
January.

\textbf{Number of \algdp calls:}
Next, we study relative efficiency when compared \algdp. Since all 4
approximation schemes use \algdp as a subroutine, a natural way of measuring
the efficiency is to study the number of \algdp calls. We report the number of
calls as a function of $\epsilon$ for datasets \dtname{Mine} and
\dtname{Crimes} in Figure~\ref{fig:calls}. Here, we did not use the speed-up version
of \algexpalpha.

\begin{figure}
\begin{tikzpicture}
\begin{axis}[xlabel={parameter $\epsilon$}, ylabel= {tests},
	title = {\scriptsize\dtname{Mine}},
    height = 3.5cm,
    width = 8.7cm,
    cycle list name=yaf,
	tick scale binop = \times,
	xticklabel = {\pgfmathparse{\tick}$2^{-\pgfmathprintnumber{\pgfmathresult}}$},
	yticklabel = {\pgfmathparse{\tick}$2^{\pgfmathprintnumber{\pgfmathresult}}$},
	xtick = {0, 1, 2, 3, 4, 5, 6, 7, 8, 9},
	ymax = 25,
	legend style = {at = {(0, 1)}, anchor = {north west}}
    ]
\addplot[yafcolor1] table[x expr = {\coordindex}, y index = 1, header = false]  {geo_tests.dat}
node[sloped, pos = 0.5, text = black, font = \scriptsize, above, inner sep = 2pt] {\alggeoalpha};

\addplot[yafcolor2] table[x expr = {\coordindex}, y index = 3, header = false]  {geo_tests.dat}
node[sloped, pos = 0.5, text = black, font = \scriptsize, above, inner sep = 2pt] {\alggeo};

\pgfplotsextra{\yafdrawaxis{0}{9}{3.3}{25}}
\end{axis}
\end{tikzpicture}

\begin{tikzpicture}
\begin{axis}[xlabel={parameter $\epsilon$}, ylabel= {tests},
	title = {\scriptsize\dtname{Crimes}},
    height = 3.5cm,
    width = 8.7cm,
    cycle list name=yaf,
	tick scale binop = \times,
	xticklabel = {\pgfmathparse{\tick}$2^{-\pgfmathprintnumber{\pgfmathresult}}$},
	yticklabel = {\pgfmathparse{\tick}$2^{\pgfmathprintnumber{\pgfmathresult}}$},
	xtick = {0, 1, 2, 3, 4, 5, 6, 7, 8, 9},
	ytick = {5, 10, 15, 20, 25},
	ymax = 28,
	legend style = {at = {(0, 1)}, anchor = {north west}}
    ]
\addplot[yafcolor1] table[x expr = {\coordindex}, y index = 1, header = false]  {exp_tests.dat}
node[sloped, pos = 0.5, text = black, font = \scriptsize, above, inner sep = 2pt] {\algexpalpha};
\addplot[yafcolor2] table[x expr = {\coordindex}, y index = 3, header = false]  {exp_tests.dat}
node[sloped, pos = 0.5, text = black, font = \scriptsize, above, inner sep = 2pt] {\algexp};

\pgfplotsextra{\yafdrawaxis{0}{9}{2.3}{28}}
\end{axis}
\end{tikzpicture}

\caption{Number of \algdp calls as a function of $\epsilon$. Both $x$ and
$y$-axis are logarithmic. We set $k = 4$, $\alpha = 0.5$ for \algexpalpha, and
$\alpha = 2$ for \algexpalpha. Here, we did not use the speed-up version of \algexpalpha.}
\label{fig:calls}

\end{figure}
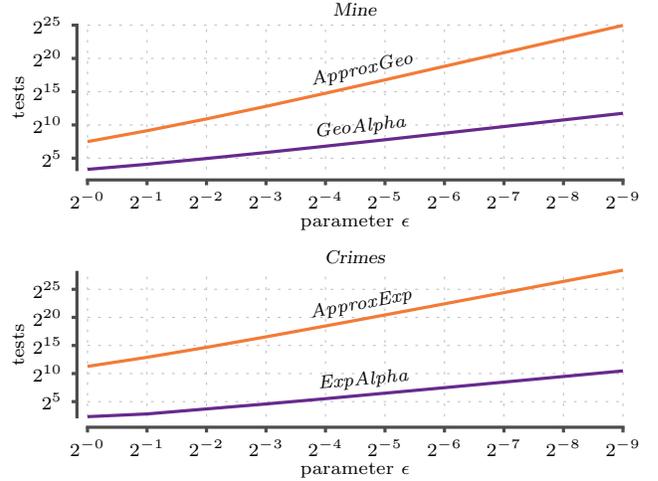

We see that the behaviour depends heavily on the accuracy parameter $\epsilon$:
for example, if we use $\epsilon = 0.5$,  then \alggeoalpha uses
17 calls while \alggeo uses 561 calls; if we set $\epsilon = 2^{-9}$, then
\alggeoalpha needs 3453 calls while \alggeo needs 32\,810\,406 calls.
This implies that we should not use extremely small $\epsilon$, especially
if we also wish to optimize $\alpha$. Nevertheless, the algorithms are
fast when we use moderately small $\epsilon$.

\textbf{Effect of a speed-up:}
Finally, we compare the effect of a speed-up for \algexpalpha
described in Section~\ref{sec:expburst}. Here we used both datasets
\dtname{Mine} and \dtname{Crimes} to which we apply \algexpalpha with $k = 5$ and $\alpha = 2$.
We vary $\epsilon$ from $2^{-13}$ to $1/2$ and compare the plain version
vs. speed-up in Figure~\ref{fig:speedup}.

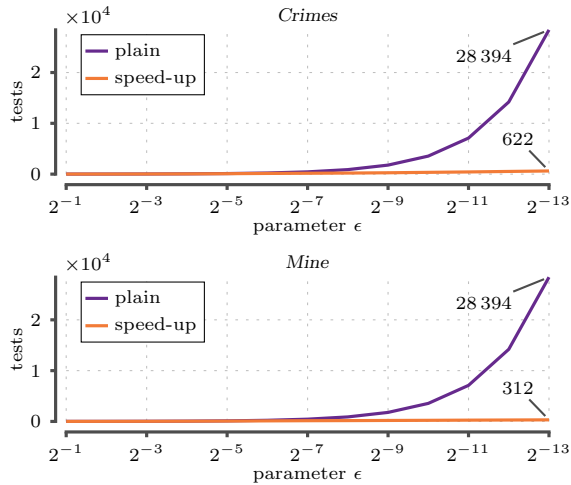
\begin{figure}[ht!]
\begin{center}
\begin{tikzpicture}
\begin{axis}[xlabel={parameter $\epsilon$}, ylabel= {tests},
	title = {\scriptsize\dtname{Crimes}},
    height = 3.5cm,
    width = 8cm,
    cycle list name=yaf,
	tick scale binop = \times,
	xticklabel = {\pgfmathparse{\tick}$2^{-\pgfmathprintnumber{\pgfmathresult}}$},
	xtick = {1, 3, 5, 7, 9, 11, 13},
	legend pos = {north west}
    ]
\addplot[yafcolor1] table[x expr = {\coordindex + 1}, y index = 1, header = false]  {speedup_narc.dat};
\addplot[yafcolor2] table[x expr = {\coordindex + 1}, y index = 2, header = false]  {speedup_narc.dat};
\legend {plain, speed-up}

\node[coordinate,pin={[inner sep = 2pt, pin edge = {thick, yafaxiscolor, shorten <=2pt}]205:{\scriptsize $28\,394$}}] at (axis cs:13, 28394) {};
\node[coordinate,pin={[inner sep = 2pt, pin distance = 3mm, pin edge = {thick, yafaxiscolor, shorten <=2pt}]115:{\scriptsize $622$}}] at (axis cs:13, 622) {};

\pgfplotsextra{\yafdrawaxis{1}{13}{9}{28394}}
\end{axis}
\end{tikzpicture}
\begin{tikzpicture}
\begin{axis}[xlabel={parameter $\epsilon$}, ylabel= {tests},
	title = {\scriptsize\dtname{Mine}},
    height = 3.5cm,
    width = 8cm,
    cycle list name=yaf,
	tick scale binop = \times,
	xticklabel = {\pgfmathparse{\tick}$2^{-\pgfmathprintnumber{\pgfmathresult}}$},
	xtick = {1, 3, 5, 7, 9, 11, 13},
	legend pos = {north west}
    ]
\addplot[yafcolor1] table[x expr = {\coordindex + 1}, y index = 1, header = false]  {speedup_mine.dat};
\addplot[yafcolor2] table[x expr = {\coordindex + 1}, y index = 2, header = false]  {speedup_mine.dat};
\legend {plain, speed-up}

\node[coordinate,pin={[inner sep = 2pt, pin edge = {thick, yafaxiscolor, shorten <=2pt}]205:{\scriptsize $28\,394$}}] at (axis cs:13, 28394) {};
\node[coordinate,pin={[inner sep = 2pt, pin distance = 3mm, pin edge = {thick, yafaxiscolor, shorten <=2pt}]115:{\scriptsize $312$}}] at (axis cs:13, 312) {};

\pgfplotsextra{\yafdrawaxis{1}{13}{9}{28394}}
\end{axis}
\end{tikzpicture}
\end{center}
\caption{Number of tests needed as a function of $\epsilon$. Speed-up (see, Section~\ref{sec:expburst}) vs. vanilla version.}
\label{fig:speedup}
\end{figure}

We see in Figure~\ref{fig:speedup} that the we gain significant speed-up as we
decrease $\epsilon$: At best, we improve by two orders of magnitude.

\section{Concluding remarks}\label{sec:conclusions}
In this paper we presented variants of~\cite{kleinberg:03:burst} for
discovering bursts: instead of deriving the base rate from $\mu$, the average
delay time between the events, we optimize this parameter along with the actual
burst discovery.
We showed that this leads to better burst discovery,
especially if the bursts are long. 
We also propose variants, where we optimize the change parameter $\alpha$,
instead of having it as a parameter.

Despite being a minor tweak, the resulting optimization problems are significantly
harder. To solve the problems, we introduce efficient algorithms yielding $(1 +
\epsilon)$ approximation guarantee. These methods are based on testing multiple
values for the base rate, and selecting the burst sequence with the best score.
Despite being similar problems, discrete and continuous versions of the problem
required their own algorithms. In addition, we were able significantly speed-up
the exponential model variant by safely ignoring some candidate values for the base rate.

The approximation algorithms are quasi-linear with respect to sequence length.
However, especially when we optimize $\alpha$, the algorithms depend also on
the actual values of the sequence, see Table~\ref{tab:algos}. A potential future work is to
improve the algorithms, and develop polynomially strong approximation schemes.
The other fruitful direction is to develop heuristics that allow us to ignore
large parts of the parameters, similar to the speed-up we propose for the
exponential model variant of the problem.

\bibliography{bibliography}

\clearpage
\appendix
\section{Viterbi algorithm for solving Problem~\ref{prb:burstexporig} or Problem~\ref{prb:burstgeoorig}}
\label{sec:app_viterbi}

We can solve Problem~\ref{prb:burstexporig} or Problem~\ref{prb:burstgeoorig}
using the standard dynamic programming algorithm by~\citet{viterbi:67:dp}.
Off-the-shelf version of this algorithm requires $\bigO{nk^2}$ time. However, we
can easily speed-up the algorithm to $\bigO{nk}$.

To see this, let us first write $o[i, j]$ to express the optimal score for the
$i$th first symbols such that the last level $\ell_i = j$.  The Viterbi algorithm
uses the fact that
\[
	o[i, j] = -\log p(s_j, \beta \alpha^j) + \min_{j'} o[i - 1, j'] + \pen{j', j}
\]
to solve the optimal sequence. Define two arrays
\[
\begin{split} 
	a[j]  & = \min_{x \geq j} o[i - 1, x] + \pen{x, j} = \min_{j' \geq j} o[i - 1, x]\quad\text{and} \\
	b[j]  & = \min_{x \leq j} o[i - 1, x] + \pen{x, j}\quad. 
\end{split}
\]
By definition, we have 
\[
	o[i, j] = -\log p(s_j, \beta \alpha^j) + \min_{j' \in \set{a[j], b[j]}} o[i - 1, j'] + \pen{j', j},
\]
that is, we can compute $o[i, j]$ in constant time as long as we have $a[j]$ and $b[j]$.
To compute $a[j]$ fast, note that either $a[j] = j$ or $a[j] = a[j + 1]$, whichever produces
better score. Similarly, due to linearity of $\pen{}$, we have $b[j] = j$ or $b[j] = b[j - 1]$,
whichever produces better score. This leads to a simple dynamic program given in Algorithm~\ref{alg:dp}
that performs in $\bigO{nk}$ time.

\begin{algorithm}[ht]
	\ForEach{$i = 1, \ldots, n$} {
		compute $a$ and $b$ in $\bigO{k}$ time\;
		\ForEach{$j = 0, \ldots, k$} {
			$c_1 \define o[i - 1, a[j]]$\;
			$c_2 \define o[i - 1, b[j]] + \pen{b[j], j}$\;
			$o[i, j] \define \min (c_1, c_2) - \log p(s_i; \beta \alpha^j)$\;
		}
		
	}
\caption{$\algdp(S, \alpha, \beta, \gamma, k, p)$, a dynamic program in order to discover burstiness}
\label{alg:dp}
\end{algorithm}

\section{Proof of Proposition~\ref{prop:boundprb}}
\label{sec:app_boundprb}

\begin{proof}
Let $L$ and $\beta$ be the solution to $\brstexpprb(\alpha)$.
Since 
\[
- \log \pexp{s ; \lambda} = s \lambda - \log \lambda,
\]
we can decompose the score $\scoreexp{L, S; \alpha, \beta, \gamma}$ as 
\[
	 \sum_{i = 1}^n \beta \alpha^{\ell_i} s_i - \log \beta - \ell_i \log \alpha + \pen{\ell_{i - 1}, \ell_i} \quad.
\]
Define $d = \sum_{i} \max\pr{\ell_i - \ell_{i - 1}, 0}$ and $m = \sum_i \ell_i$.
Let us also write $f(L) =  \sum_{i} \alpha^{\ell_i} s_i$.
Then the score becomes
\begin{equation}
	 \beta f(L) - n\log \beta - m \log \alpha + d \gamma \log n\quad.
\label{eq:score}
\end{equation}
Obviously, $L$ satisfies the constraints posed in \boundprb.
Moreover, $L$ minimizes $f(L)$ (within the constraints); otherwise we could
replace $L$ with $L'$, making the first term in Eq.~\ref{eq:score} genuinely smaller and keeping the remaining
terms constant. This contradicts the optimality of $L$. Consequently, $L$ solves
\boundprb.

To prove the remaining claims, first note that $\beta$ optimizing Eq.~\ref{eq:score}
must satisfy 
\[
	\frac{\partial \scoreexp{}}{\partial \beta} =  f(L) - n / \beta = 0,
\]
proving the claim regarding $\beta$.

Since $\ell_i \leq k$, we have $m \leq nk$.
To bound $d$, let us write $e = \sum_i \max\pr{\ell_{i - 1} - \ell_{i}, 0}$.
We have
\[
	d - e = \sum_{i = 1}^n \ell_i - \ell_{i - 1} = \ell_n \leq k 
\]
and
\[
	d + e = \sum_{i = 1}^n \abs{\ell_i - \ell_{i - 1}} \leq nk\quad.
\]
Summing the inequalities leads to $2d \leq (n + 1)k$,
which proves the proposition.
\end{proof}

\section{Proof of Proposition~\ref{prop:geoapproxalpha}}
\label{sec:app_geoapproxalpha}

To prove the proposition we need several lemmas.
Throughout this section, we assume that we are given an integer delay sequence $S$,
and parameters $\alpha$, $\gamma$, $k$, and $\epsilon > 0$. We will write $\mu = \frac{1}{n}\sum_{i} s_i$ and $c = 1 + \epsilon$.

The first lemma states that the optimal $\beta$ will be between the range that \alggeoalpha tests.

\begin{lemma}
\label{lem:geobounds}
Let $L$ and $\beta$ be the solution of $\brstgeoprb(\alpha)$.
Then
	$\frac{\mu}{1 + \mu} \leq \beta \leq \frac{\mu}{1/n + \mu}$.
\end{lemma}

\begin{proof}
Since $\beta$ is optimal we must have
$\partial\scoregeo{S, L; \beta} / \partial \beta = 0$.
This implies that
\[
	\sum_{i = 1}^n \frac{s_i}{\beta}  = \sum_{i = 1}^n  \frac{1}{1 - \beta \alpha^{\ell_i}}
\quad\text{or}\quad
	\mu  = \frac{1}{n} \sum_{i = 1}^n  \frac{\beta}{1 - \beta \alpha^{\ell_i}}\quad.
\]
Since $\alpha^{\ell_i} \leq 1$ we must have
\[
	\mu =  \frac{1}{n} \sum_{i = 1}^n  \frac{\beta}{1 - \beta \alpha^{\ell_i}} \leq  \frac{\beta}{1 - \beta},
\]
which can be rewritten as $\mu / (1 + \mu) \leq \beta$. This gives us the lower bound of
the lemma.

To prove the other bound, note that we must have at least one $\ell_i = 0$.
This leads to
\[
	\mu =  \frac{1}{n} \sum_{i = 1}^n  \frac{\beta}{1 - \beta \alpha^{\ell_i}} \geq  \frac{1}{n}\frac{\beta}{1 - \beta},
\]
which can be rewritten as
$\mu / (n^{-1} + \mu) \geq \beta$. This proves the upper bound of the lemma.
\end{proof}

Next we show that if we vary $\alpha$ and $\beta$ by little, while keeping $L$
constant, the score will not change a lot.

\begin{lemma}
\label{lem:geoapprox}
Let $\alpha'$ such that $\alpha^{1 + \epsilon} \leq \alpha' \leq \alpha$ and
let $\beta'$ such that $\beta^{1 + \epsilon} \leq \beta' \leq \beta$.
Then
\[
     \scoregeo{L; \alpha', \beta'} \leq  (1 + \epsilon)\scoregeo{L; \alpha, \beta}\quad.
\]
\end{lemma}

\begin{proof}
We can decompose the score $ \scoregeo{L; \alpha, \beta}$ as
\[
\begin{split}
	&\sum_{i = 1}^n - s_i \log \beta - s_i\ell_i \log \alpha - \log (1 - \beta \alpha^{\ell_i}) + \pen{\ell_{i - 1}, \ell_i} \\
	& \qquad = D\log \beta + E \log \alpha + C(\alpha, \beta),\\
\end{split}
\]
where $C(\alpha, \beta)$ is the sum of the last two terms, $D = -\sum_{i = 1}^n s_i$, 
and $E = -\sum_{i = 1}^n \ell_i s_i$.
Note that $C(\alpha, \beta) \geq 0$ and $C(\alpha, \beta)$ increases
as a function of $\alpha$ and $\beta$. We can now upper bound the score
\[
\begin{split}
	\scoregeo{L; \alpha', \beta'}
	& =   D\log \beta' + E \log \alpha' + C(\alpha', \beta') \\
	& \leq  D\log \beta^c  + E \log \alpha^c + C(\alpha', \beta') \\
	& = c D\log \beta + c E\log \alpha + C(\alpha', \beta')  \\
	& \leq c D\log \beta +  c E\log \alpha + cC(\alpha, \beta) \\
	& = c\scoregeo{L; \alpha, \beta}\quad.
\end{split}
\]
Since $c = 1 + \epsilon$, this completes the proof.
\end{proof}

We can now prove the main result.
\begin{proof}[Proof of Prop.~\ref{prop:geoapproxalpha}]
Lemma~\ref{lem:geobounds} guarantees that \alggeoalpha tests $\beta'$ such that
$(\beta^*)^{1 + \epsilon} \leq \beta' \leq \beta^*$. Let $L'$ be the optimal solution for $\beta'$.

Lemma~\ref{lem:geoapprox} guarantees that
$\scoregeo{S, L^*; \beta'} \leq  (1 + \epsilon)\scoregeo{S, L^*; \beta^*}$.
Since $\scoregeo{S, L; \beta} \leq \scoregeo{S, L'; \beta'} \leq \scoregeo{S, L^*; \beta'}$, the result follows.
\end{proof}

\section{Proof of Proposition~\ref{prop:geotimealpha}}
\label{sec:app_geotimealpha}

In order to prove the proposition we need
two lemmas. The first lemma is a technical result that is needed to prove the
second lemma.

\begin{lemma}
\label{lem:difffrac}
Define 
\[
    h(x, y) = \log \log \frac{y + x}{x}\quad.
\]
Then
\[
	h(x_1, y_2) - h(x_1, y_1) \geq h(x_2, y_2) - h(x_2, y_1)
\]
for $x_1 \leq x_2$ and $y_1 \leq y_2$.
\end{lemma}

\begin{proof}
The partial derivative of $h$ is equal to
\[
	\frac{\partial h(x, y)}{\partial x} = \spr{\frac{1}{y + x} - \frac{1}{x}} \frac{1}{\log(y + x) - \log x},
\]
and it is decreasing as a function of $y$.
This implies that
\[
	\frac{\partial h(x, y_2)}{\partial x} - \frac{\partial h(x, y_1)}{\partial x} \leq 0,
\]
that is $h(x, y_2) - h(x, y_1)$ is decreasing as a function of $x$,
which proves the lemma.
\end{proof}

Our second lemma essentially shows that \alggeoalpha does not test too many values. 

\begin{lemma}
\label{lem:georounds}
Let $n$ be an integer, and let $\mu \geq 1 / n$ be a real number. 
Assume $\epsilon > 0$ and let $c = 1 + \epsilon$. Let $r$ be such that
\[
	\beta^{(c^r)} \geq \frac{\mu}{1 + \mu}, \quad{where}\quad
	\beta = \frac{\mu}{1/n + \mu}
	\quad.
\]
Then
\[
	r \leq  \frac{\log \log (n + 1) - \log \log 2}{\log c} \in \bigO{\frac{\log \log n}{\epsilon}}\quad.
\]
\end{lemma}
\begin{proof}
We begin by applying $-\log$ to the inequality $\beta^{(c^r)} \geq \frac{\mu}{1 + \mu}$ and obtain
\[
	c^r \log \frac{1/n + \mu}{\mu} \leq \log \frac{1 + \mu}{\mu}\quad.
\]
Another application of $\log$ and using $h$, as defined in Lemma~\ref{lem:difffrac},
leads us to
\[
	r\log c \leq h(\mu, 1) - h(\mu, 1/n)\quad.
\]
Since $\mu \geq 1/n$, Lemma~\ref{lem:difffrac} implies
\[
\begin{split}
	r  \log c & \leq h(\mu, 1) - h(\mu, 1/n) \\
	 & \leq h(1/n, 1) - h(1/n, 1/n) \\
	 & = \log \log \frac{1 + 1/n}{1/n} - \log \log \frac{2/n}{1/n}  \\
	 & = \log \log (n + 1) - \log \log 2,  \\
\end{split}
\]
which gives us the needed inequality.
Since $1 / \log c \leq \frac{1 + \epsilon}{\epsilon} \in \bigO{1 / \epsilon}$, the result follows.
\end{proof}

We can now prove the main result.
\begin{proof}[Proof of Prop.~\ref{prop:geotimealpha}]
Lemma~\ref{lem:georounds} guarantees that we only test $\bigO{\epsilon^{-1} \log \log n}$ values of $\beta$.
Each test requires $\bigO{nk}$ time, which proves the result.
\end{proof}

\section{Proof of Proposition~\ref{prop:geoapprox}}
\label{sec:app_geoapprox}

Assume a delay sequence $S$, parameter $\gamma$,
and an upper bound for levels $k$.
Let $\alpha$, $\beta$ and $L$ be the solution to \brstgeoprb.
Let $\mu$ be the average of the $S$.
Write $\Delta_i = \pen{\ell_{i - 1}, \ell_i}$.
Assume that we are given $\epsilon > 0$ and let $c = 1 + \epsilon$.

First, we upper-bound the search space for $\alpha$.

\begin{lemma}
\label{lem:geoapproxupper}
Let $\sigma = \frac{\mu}{1/n + \mu}$ and $\alpha' = \sigma^{\epsilon / k}$.
If $\alpha \geq \alpha'$, then
$\scoregeo{L; \alpha', \beta} \leq
(1 + \epsilon)\scoregeo{L; \alpha, \beta}$
\end{lemma}

\begin{proof}
Decompose the score $ \scoregeo{L}$ to
\[
	\sum_{i = 1}^n - s_i \log \beta - s_i\ell_i \log \alpha - \log (1 - \beta \alpha^{\ell_i}) + \Delta_i\quad. 
\]
Let $C(\alpha)$ be the sum of the last two terms.
Since $C(\alpha) \geq C(\alpha') \geq 0$, we only need to show that
\[
	-s_i \log \beta - s_i \ell_i \log \alpha' \leq -cs_i \log \beta - cs_i \ell_i \log \alpha\quad.
\]
To show this, note that Lemma~\ref{lem:geobounds} implies that $\beta \leq \sigma$. We can now bound the first two terms by
\[
\begin{split}
	-s_i \log \beta - s_i \ell_i \log \alpha' & = -s_i \log \beta - \epsilon s_i \frac{\ell_i}{k} \log \sigma \\
	& \leq -s_i \log \beta - \epsilon s_i \frac{\ell_i}{k} \log \beta \\
	& \leq -s_i \log \beta - \epsilon s_i \log \beta \\
	& = -cs_i \log \beta \\
	& \leq -cs_i \log \beta - cs_i \ell_i \log \alpha\quad.\\
\end{split}
\]
This proves the lemma.
\end{proof}

Next, we lower-bound the search space for $\alpha$.

\begin{lemma}
\label{lem:geoapproxlower}
If there is an index $j$ such that $s_jl_j \neq 0$,
then $\alpha \geq  \frac{1}{1 + nk}$.
\end{lemma}

\begin{proof}
In order for $\alpha$ to be optimal, $\frac{\partial \scoregeo{}}{\partial \alpha} = 0$, or
\[
	\sum_{i = 1}^n s_i \ell_i  = \sum_{i = 1}^n \frac{\beta \ell_i \alpha^{\ell_i}}{1 - \beta \alpha^{\ell_i}} = \sum_{i = 1}^n f(\alpha, \beta, l_i),
\]

We can upper bound the right-hand side.
If $\ell_i = 0$, then
\[
	\frac{\beta \ell_i \alpha^{\ell_i}}{1 - \beta \alpha^{\ell_i}} = 0 \leq \frac{k\alpha}{1 - \alpha}\quad.
\]
If $\ell_i \geq 1$, then
\[
	\frac{\beta \ell_i \alpha^{\ell_i}}{1 - \beta \alpha^{\ell_i}}
	\leq \ell_i \frac{\beta  \alpha}{1 - \beta \alpha}
	\leq k \frac{\beta  \alpha}{1 - \beta \alpha}
	\leq k \frac{\alpha}{1 - \alpha}\quad.
\]

This leads to
\[
	\sum_{i = 1}^n s_i \ell_i \leq kn \frac{\alpha}{1 - \alpha}.
\]
Since $s_j l_j \geq 1$, we have
\[
	\frac{\alpha}{1 - \alpha} \geq  \frac{1}{nk} \sum_{i = 1}^n s_i \ell_i\geq \frac{1}{nk},
\]
which can be rewritten as $\alpha \geq \frac{1}{1 + nk}$.
\end{proof}

The next lemma addressed the case when the condition of the previous lemma fails.

\begin{lemma}
\label{lem:special}
If $s_il_i = 0$ for all $i$, then $\alpha = 0$.
\end{lemma}

\begin{proof}
We can decompose the score $ \scoregeo{L}$ as
\[
	\sum_{i = 1}^n - s_i \log \beta - \log (1 - \beta \alpha^{\ell_i}) + \Delta_i\quad. 
\]
This score decreases as a function of $\alpha$, and is minimized when $\alpha = 0$.
\end{proof}

We can now prove the main result.
\begin{proof}[Proof of Prop.~\ref{prop:geoapprox}]
If the condition in Lemma~\ref{lem:special} is triggered, then $\alpha = 0$,
which is tested by \alggeo, and Proposition~\ref{prop:geoapproxalpha} guarantees the result.

Otherwise, Lemmas~\ref{lem:geobounds},~\ref{lem:geoapproxupper}~and~\ref{lem:geoapproxlower}
guarantee that \alggeo tests $\alpha'$ and $\beta'$ such that
$(\beta^*)^{1 + \epsilon} \leq \beta' \leq \beta^*$ and
$(\alpha^*)^{1 + \epsilon} \leq \alpha' \leq \alpha^*$.
Let $L'$ be the optimal solution for $\alpha'$, $\beta'$.

The proposition follows since
\[
\scoregeo{L'; \alpha', \beta'} \leq
\scoregeo{L^*; \alpha', \beta'} \leq  c\scoregeo{L^*; \alpha^*, \beta^*},
\]
where the second inequality is due to Lemma~\ref{lem:geoapprox}. 
\end{proof}

\section{Proof of Proposition~\ref{prop:geotime}}
\label{sec:app_geotime}

\begin{proof}
Let $\sigma = \mu / (1/n + \mu)$. 
Let $m$ be the number of tests for different $\alpha$s.
The stopping condition now guarantees
\[
	\pr{\frac{1}{1 + nk}}^{c^{-m}} \leq \sigma^{\epsilon / k},
\]
which can be rewritten as
\[
	-c^{-m} \log (1 + nk) \leq \frac{\epsilon}{k}\log \sigma \leq \frac{\epsilon}{k}(\sigma - 1) = \frac{-\epsilon}{k(1 +n\mu)}
\]
Reversing the sign, and taking logarithm leads to
\[
	-m \log c \geq \log(\epsilon) -\log(k) -\log(1 + n\mu) - \log \log (1 + nk)
\]
which leads to
\[
\begin{split}
	m & \leq \frac{ \log(k) + \log(1 + n\mu) -\log (\epsilon) + \log \log (1 + nk)}{\log c} \\
	& \in \bigO{\frac{\log n + \log \mu + \log k - \log\epsilon}{\epsilon}}\quad.
\end{split}
\]
\end{proof}

\section{Proof of Proposition~\ref{prop:expapproxalpha}}
\label{sec:app_expapproxalpha}

To prove the proposition we need several lemmas.
Throughout this section, we assume that we are given a delay sequence $S$,
and parameters $\alpha$, $\gamma$ and $k$. We will write $\mu = \frac{1}{n}\sum_{i} s_i$.

First,  we need to show that the optimal $\beta$ stays within the bounds used
by \algexp.

\begin{lemma}
\label{lem:expbounds}
Let $\beta$ and $L$ be the solution to $\brstexpprb(\alpha)$. Then
	$\frac{1}{\alpha^k \mu} \leq \beta \leq \frac{1}{\mu}$.
\end{lemma}

\begin{proof}
Write $\Delta_i = \pen{\ell_{i - 1}, \ell_i}$.
We can decompose the score $\scoreexp{L; \beta}$ as
\[
	\sum_{i = 1}^n \beta \alpha^{\ell_i} s_i - \log \beta - \ell_i \log \alpha + \Delta_i \quad.
\]
Let $g = \frac{1}{n}\sum_{i} \alpha^{\ell_i} s_i$.
Due to optimality of $\beta$ we must have
\[
	\frac{\partial \scoreexp{}}{\partial \beta} = ng - \frac{n}{\beta} = 0,
\]
that is, $\beta = 1/g$.
As $0 \leq \ell_i  \leq k$, we have $\mu \leq g \leq \alpha^k \mu$.
This proves the lemma.
\end{proof}

Our next result is a technical lemma that is needed to control the possible
negative terms in the score.

\begin{lemma}
\label{lem:geomean}
Let $\beta$ and $L$ be the solution to \brstexpprb.
Let $g = \spr{\prod_{i} s_i}^{1/n}$
be the geometric mean. Then
\[
	\sum_{i = 1}^n -\log \beta - \ell_i \log \alpha \geq n \log g\quad.
\]
\end{lemma}

\begin{proof}
Let $f = \frac{1}{n}\sum_{i} \alpha^{\ell_i} s_i$ and
$h = \spr{\prod_{i} \alpha^{\ell_i}}^{1/n}$. The arithmetic-geometric mean
inequality states that $hg \leq f$.
By definition, we must have $\beta = 1/f$. This leads to
\[
\begin{split}
	\sum_{i = 1}^n -\log \beta - \ell_i \log \alpha & = n \log f - \sum_{i = 1}^n \ell_i \log \alpha \\
	& \geq n \log gh - \sum_{i = 1}^n \ell_i \log \alpha \\
	& = n \log g + n \log h - \sum_{i = 1}^n \ell_i \log \alpha \\
	& = n \log g \quad. \\
\end{split}
\]
This proves the lemma.
\end{proof}

The next lemma shows that if we vary $\beta$ by little while keeping $L$
constant, the score of the solution will not change a lot.

\begin{lemma}
\label{lem:expapprox}
Let $\beta$ and $L$ be the solution to \brstexpprb.
Let $g = \spr{\prod_{i} s_i}^{1/n}$ be the geometric mean, and let $\psi = n \log g$. 
Let $\epsilon_1, \epsilon_2 > 0$ and assume $\alpha'$, $\beta'$ such that
$\beta \leq \beta' \leq \beta(1 + \epsilon_1)$ and
$\alpha \leq \alpha' \leq \alpha(1 + \epsilon_2)$.
Then
\[
	\scoreexp{L; \alpha', \beta'} - \psi \leq  c(\scoreexp{L; \alpha, \beta} - \psi),
\]
where $c = (1 + \epsilon_1)(1 + \epsilon_2)^k$.
\end{lemma}

\begin{proof}
Write $\Delta_i = \pen{\ell_{i - 1}, \ell_i}$.
Decompose the score $\scoreexp{L; \alpha, \beta}$ to
\[
	\sum_{i = 1}^n \beta \alpha^{\ell_i} s_i - \log \beta - \ell_i \log \alpha + \Delta_i,
\]
and let $C(\alpha, \beta)$ be the sum of last three terms.
Note that $C(\alpha', \beta') \leq
C(\alpha, \beta)$, and due to Lemma~\ref{lem:geomean} $C(\alpha, \beta) \geq \psi$.
Also let $F(\alpha) = \sum_{i = 1}^n \alpha^{\ell_i} s_i$ be the sum of the first term without $\beta$.
We can now write
\[
\begin{split}
	\scoreexp{L; \alpha', \beta'} - \psi & =  \beta'F(\alpha') + C(\alpha', \beta') - \psi  \\
	& \leq c\beta F(\alpha) + C(\alpha, \beta)  - \psi\\
	& \leq c\beta F(\alpha) + c(C(\alpha, \beta) - \psi) \\
	&= c(\scoreexp{L; \alpha, \beta} - \psi),  \\
\end{split}
\]
which proves the lemma.
\end{proof}

We can now prove the main result.
\begin{proof}[Proof of Prop.~\ref{prop:expapproxalpha}]
Lemma~\ref{lem:expbounds} guarantees that \algexp tests $\beta'$ such that
$\beta^* \leq \beta' \leq \beta^*(1 + \epsilon)$. Let $L'$ be the optimal solution for $\beta'$.

Lemma~\ref{lem:expapprox} guarantees that
$\scoreexp{L^*; \beta'} - n \log g \leq  (1 + \epsilon)(\scoreexp{L^*; \beta^*} - n \log g)$.
Since $\scoreexp{L; \beta} \leq \scoreexp{L'; \beta'} \leq \scoreexp{L^*; \beta'}$, the result follows.
\end{proof}

\section{Proof of Proposition~\ref{prop:exptimealpha}}
\label{sec:app_exptimealpha}

\begin{proof}
Assume that $(1 + \epsilon)^{-r}/\mu \geq 1 / (\alpha^k \mu)$.
Solving for $r$ gives us 
\[
	r \leq \frac{k \log \alpha }{ \log (1 + \epsilon)} \in \bigO{\epsilon^{-1} k \log \alpha}\quad.
\]
Consequently, \algexp has at most $\bigO{\epsilon^{-1} k \log \alpha}$ iterations. Since a single iteration
costs $\bigO{nk}$ time, the result follows.
\end{proof}

\section{Proof of Proposition~\ref{prop:expapprox}}
\label{sec:app_expapprox}
We first upper-bound the optimal $\alpha$.

\begin{lemma}
\label{lem:expalphaupper}
Let $\Omega = \max S$ and $\omega = \min S$.
Then $\alpha \leq \Omega / \omega$.
\end{lemma}

\begin{proof}
Assume that $\alpha > \Omega / \omega$. 
To prove the result we use the fact that
\begin{equation}
\label{eq:lowerimprove}
	s_i \lambda' - \log \lambda' \leq s_i \lambda - \log \lambda,
\end{equation}
when $s_i \geq \lambda' \geq \lambda$
or $s_i \leq \lambda' \leq \lambda$.

We claim that $\beta \alpha \geq \omega$. Assume otherwise.
Consider an alternative level sequence $\ell'_i = \max(\ell_i - 1, 0)$
and $\beta' = \beta \alpha$.
Under this transformation, the only modelling terms
in $\pexp{}$
that change are the original levels for which $\ell_i = 0$,
equal to
\[
	(\beta' s_i - \log \beta') - (\beta s_i - \log \beta)\quad.
\]
Eq.~\ref{eq:lowerimprove} guarantees that this change is always negative.
In addition, $\Delta_i$ can only decrease. We can repeat this argument until
$\beta \alpha \geq \omega$.

We now split in two separate cases.

Case (\emph{i}):
Assume $\beta \alpha \geq \Omega$. We must have $\beta \alpha = \Omega$.
Otherwise, since $\beta \alpha > s_i$, Eq.~\ref{eq:lowerimprove} now guarantees
that we can safely decrease $\alpha$ at least until $\beta \alpha = \Omega$.
The assumption $\alpha > \Omega / \omega$ implies that $\beta \leq \omega$.
Similarly, we must have $\beta = \omega$.
Otherwise, if we increase $\beta$ and decrease $\alpha$
such that $\beta\alpha$ remains constant, then Eq.~\ref{eq:lowerimprove}
implies that the score decreases until $\beta = \omega$. Consequently, $\alpha = \Omega / \omega$ which
contradicts the assumption $\alpha > \Omega / \omega$.

Case (\emph{ii}):
Assume $\beta \alpha \leq \Omega$. In other words we have
$\beta \leq \omega \leq \beta \alpha \leq \Omega \leq \beta \alpha^2$.
We claim that $\beta = \omega$ or $\Omega = \beta \alpha^2$.
Otherwise, increase $\beta$ and decrease $\alpha$ such that $\beta\alpha$ remains constant;
Eq.~\ref{eq:lowerimprove} states that we can only decrease
the score at least until $\beta = \omega$ or $\Omega = \beta \alpha^2$.  
This immediately implies $\alpha \leq \Omega / \omega$, which is a contradiction.
\end{proof}

\begin{proof}[Proof of Prop.~\ref{prop:expapprox}]
Lemma~\ref{lem:expalphaupper} guarantees that \algexp
tests $\alpha'$ such that $\alpha \leq \alpha' \leq \alpha\sqrt[2k]{1 + \epsilon}$.
Lemma~\ref{lem:expapprox} now guarantees the approximation ratio of
\[
	(1 + \epsilon / 2) \sqrt[2k]{1 + \epsilon}^k \leq (1 + \epsilon),
\]
which proves the result.
\end{proof}

\section{Proof of Proposition~\ref{prop:exptime}}
\label{sec:app_exptime}

\begin{proof}
The number of iterations done by \algexp is
	$\bigO{\epsilon^{-1}k(\log \Omega - \log \omega)}$.
A single iteration requires $\bigO{\epsilon^{-1} nk^2 \log \alpha} \in \bigO{\epsilon^{-1} nk^2 \log(\Omega / \omega)}$ time.
\end{proof}

\section{Proof of Proposition~\ref{prop:speedup}}

Let us write $\lambda(\beta)$ to be an optimal solution using $\beta$ as a parameter.
Define
\[
	f(L) = \sum_{i = 1}^n s_i \alpha^{\ell_i},
\]
and let $h(\beta) = n / f(\lambda(\beta))$.
Note that we may have several optimal solutions for $\lambda(\beta)$ and they
may yield different values of $f(L)$.  We break the ties
with a lexicographical order.

We will first prove that $h$ is monotonic.

\begin{lemma} 
\label{lem:monotone}
$h(\beta_1) \geq h(\beta_2)$ for $\beta_1 \leq \beta_2$.
\end{lemma}

\begin{proof}
Similar to Eq.~\ref{eq:score}, we can decompose the score $\scoreexp{L, S; \alpha, \beta, \gamma} = \beta f(L) - n \log \beta + C(L)$,
where
\[
	C(L) = \sum_{i = 1}^n  - \ell_i \log \alpha + \pen{\ell_{i - 1}, \ell_i}\quad.
\]
We will suppress $S$, $\alpha$, and $\gamma$ from the notation since they are constant.

Assume that we have two level sequences $L_1$ and $L_2$. Note that if 
\[
	\beta f(L_1) - n \log \beta + C(L_1) =  \beta f(L_2) - n \log \beta + C(L_2)
\]
then
\[
	\beta f(L_1) + C(L_1) =  \beta f(L_2) + C(L_2)\quad.
\]
This implies that there are three possible cases:
\emph{(i)}
$L_1$ and $L_2$ yield exact same cost for every $\beta$,
\emph{(ii)}
the cost for $L_1$ is always smaller than the cost for $L_2$, or vice versa, or
\emph{(iii)}
there is exactly one parameter, say $\beta'$, where
$\scoreexp{L_1; \beta'} = \scoreexp{L_2; \beta'}$.
In the first case, we must have $f(L_1) = f(L_2)$. 
In the last case, we must have $f(L_1) \neq f(L_2)$.
Assume that $f(L_1) < f(L_2)$. Then for every $\beta \leq \beta'$,
\[
\begin{split}
	&\scoreexp{L_1; \beta} - \scoreexp{L_2; \beta} \\
	& = \beta (f(L_1) - f(L_2)) + C(L_1) - C(L_2) \\
	& \geq \beta' (f(L_1) - f(L_2)) + C(L_1) - C(L_2) \\
	& = \scoreexp{L_1; \beta'} - \scoreexp{L_2; \beta'} = 0, \\
\end{split}
\]
and similarly $\scoreexp{L_1; \beta} \leq \scoreexp{L_2; \beta}$, for every $\beta \geq \beta'$.

Let $\beta_1 < \beta_2$, and let $L_1 = \lambda(\beta_1)$ and $L_2 = \lambda(\beta_2)$.
There are four possible cases.

Case (\emph{a}):
If $\scoreexp{L_1; \beta_1} = \scoreexp{L_2; \beta_1}$ and $\scoreexp{L_1; \beta_2} = \scoreexp{L_2; \beta_2}$,
then Case \emph{(i)} guarantees that $f(L_1) = f(L_2)$.

Case (\emph{b}):
If $\scoreexp{L_1; \beta_1} = \scoreexp{L_2; \beta_1}$ and $\scoreexp{L_1; \beta_2} > \scoreexp{L_2; \beta_2}$,
then Case \emph{(iii)} guarantees that $\beta_1 = \beta'$ and $f(L_1) \geq f(L_2)$.

Case (\emph{c}):
If $\scoreexp{L_1; \beta_1} < \scoreexp{L_2; \beta_1}$ and $\scoreexp{L_1; \beta_2} = \scoreexp{L_2; \beta_2}$,
then Case \emph{(iii)} guarantees that $\beta_2 = \beta'$ and $f(L_1) \geq f(L_2)$.

Case (\emph{d}):
If $\scoreexp{L_1; \beta_1} < \scoreexp{L_2; \beta_1}$ and $\scoreexp{L_1; \beta_2} > \scoreexp{L_2; \beta_2}$,
then Case \emph{(iii)} guarantees that $\beta_1 < \beta' < \beta_2$ and $f(L_1) \geq f(L_2)$.

These cases immediately guarantee that $h(\beta_1) = n / f(L_1) \leq n / f(L_2) = h(\beta_2).$
\end{proof}

\begin{proof}[Proof of Proposition~\ref{prop:speedup}]
Since $\beta^*$ is optimal we must have $\beta^* = h(\beta^*)$.
Assume that $\beta \leq \beta^*$.
Then Lemma~\ref{lem:monotone} states that $\beta' = h(\beta) \leq h(\beta^*) = \beta^*$.
Assume that $\beta^* \leq \beta$.
Then Lemma~\ref{lem:monotone} states that $\beta^* = h(\beta^*) \leq h(\beta) = \beta'$.
\end{proof}

\end{document}